\documentclass{scrartcl}

\usepackage[draft,inline]{fixme}
\usepackage{array}
\usepackage{multirow}

\usepackage{mathrsfs}

\usepackage{natbib} 
\usepackage{stmaryrd}
\usepackage{hyperref}
\hypersetup{colorlinks=true, citecolor=black, linkcolor=blue, urlcolor=blue}  

\usepackage{enumitem}

\usepackage{tikz}
\usepackage{pgf} 
\usetikzlibrary{patterns,automata,arrows,shapes,snakes,topaths,trees,backgrounds,positioning,through,calc}

\usepackage{amssymb,amsthm,amsmath,amsfonts}
\usepackage[framemethod=TikZ]{mdframed}

\mdfsetup{skipabove=\topskip,skipbelow=\topskip}
\newmdtheoremenv[roundcorner=8pt,backgroundcolor=gray!32,nobreak=true,linecolor=blue]{myremark}[thm]{Remark}%\begin{rem}\label{#1}}{\end{rem}\end{remarkaux}}%\ignorespacesafterend}

\FXRegisterAuthor{wh}{awh}{WH}%Wes
\FXRegisterAuthor{tl}{atl}{TL}%Tadeusz

\usepackage{mathpartir}

\title{Complete Additivity and \\ Modal Incompleteness}
\author{Wesley H. Holliday\footnote{University of California, Berkeley} \and Tadeusz Litak\footnote{Friedrich-Alexander-Universit\"{a}t Erlangen-N\"{u}rnberg}}
\date{{\normalsize Version of June 15, 2019 \\ Forthcoming in \textit{The Review of Symbolic Logic} {\small (submitted September 29, 2016)}}}                                           

\usepackage{thmtools}

\theoremstyle{definition}
\newtheorem{theo}{Theorem}[section]
\newtheorem{definition}[theo]{Definition}
\newtheorem{lemma}[theo]{Lemma}
\newtheorem{corollary}[theo]{Corollary}
\newtheorem{proposition}[theo]{Proposition}
\newtheorem{remark}[theo]{Remark}

\newcommand{\clofr}[1]{\mathcal{#1}}  %class of frames
\newcommand{\clfr}[1]{\ensuremath{\mathcal{#1}}}  %class of frames, for text mode
\newcommand{\nml}[1]{\mathsf{#1}} % notation for normal modal logic

\newcommand{\deq}{:=}

\newcommand{\vblog}{\mathsf{vB}} % for van Benthem's logic

\newcommand{\cfV}{\ensuremath{\clofr{V}}}

\newcommand{\cfA}{\ensuremath{\clofr{A}}}

\newcommand{\cfC}{\ensuremath{\clofr{C}}}

\newcommand{\cfT}{\ensuremath{\clofr{T}}}

\newcommand{\ubox}{\mathsf{A}}
\newcommand{\udiam}{\mathsf{E}}

\newcommand{\bao}{\ensuremath{\textsc{bao}}}
\newcommand{\ma}{\ensuremath{\textsc{ma}}}
\newcommand{\baos}{\ensuremath{\textsc{bao}}s}
\newcommand{\mas}{\ensuremath{\textsc{ma}}s}

\newcommand{\Ag}{\mathfrak{A}}
\newcommand{\Bg}{\mathfrak{B}}
\newcommand{\Cg}{\mathfrak{C}}

\newcommand{\goth}[1]{\mathfrak{#1}}

\newcommand{\Ff}{\ffr}
\newcommand{\Gf}{\gfr}

\newcommand{\llimp}{\,\Rightarrow\,}
\newcommand{\lland}{\,\&\,}
\newcommand{\llor}{\,\textsc{OR}\,}
\newcommand{\llnot}{\,\textsc{NOT}\,}

\newcommand{\refeq}[1]{(\ref{#1})}
\newcommand{\blid}[1]{\preceq#1} %botomless (principal) ideal

\newcommand{\resoper}{\mathbf{p}}
\newcommand{\resdual}{\mathbf{h}}

\newcommand{\nmbx}{\Box}
\newcommand{\nmdm}{\Diamond}
\newcommand{\aubx}{[e]}
\newcommand{\audm}{\langle e\rangle}

\newcommand{\numbx}[1]{[#1]}%\boxed{#1}\,}
\newcommand{\numdm}[1]{\langle#1\rangle}
\newcommand{\zebx}{\numbx{0}}
\newcommand{\zedm}{\numdm{0}}
\newcommand{\onbx}{\numbx{1}}
\newcommand{\ondm}{\numdm{1}}

\newcommand{\admis}[1]{\mathbb{#1}}

\newcommand{\vblogic}{\vblog}
\newcommand{\doubled}[1]{{#1}^\bullet}
\newcommand{\equals}{:=}
\newcommand{\dualrel}{\Box_{\alpha_1}}
\newcommand{\operrel}{\Diamond_{\alpha_1}}
\newcommand{\oper}{\Diamond}
\newcommand{\dual}{\Box}
\newcommand{\univbox}{\boxdot}
\newcommand{\univdmd}{\Diamond\hspace{-.106in}\cdot{}}
\newcommand{\imply}{\to}
\newcommand{\satisf}{\vDash}

\newcommand{\csfnt}[1]{(\textsf{#1})}
\newcommand{\mings}{\csfnt{mings}}
\newcommand{\mincc}{\csfnt{mincc}}
\newcommand{\single}{\csfnt{single}}
\newcommand{\singirr}{\ensuremath{\textsf{single}^\bullet}}
\newcommand{\singref}{\ensuremath{\textsf{single}^\circ}}

\newcommand{\fcons}{\vDash_F}

\newcommand{\frafo}[1]{\ensuremath{\mathscr{#1}}}
\newcommand{\vbfr}{\frafo{V\!B}}
\newcommand{\vbefr}{\frafo{V\!BE}}
\newcommand{\ffr}{\frafo{F}}
\newcommand{\gfr}{\frafo{G}}

\newcommand{\takeout}[1]{}

\newcommand{\vbenthem}{\vbfr}%\mathcal{VB}}

\begin{document}
\maketitle
\begin{abstract} In this paper, we tell a story about incompleteness in modal logic. The story weaves together a paper of van Benthem \citeyearpar{Benthem1979}, ``Syntactic aspects of modal incompleteness theorems,'' and a longstanding open question: whether every normal modal logic can be characterized by a class of \textit{\mbox{completely} additive} modal algebras, or as we call them, \cfV-\baos. Using a first-order reformulation of the property of complete additivity, we prove that the modal logic that starred in van Benthem's paper resolves the open question in the negative. In addition, for the case of bimodal logic, we show that there is a naturally occurring logic that is incomplete with respect to \cfV-\baos, namely the provability logic $\nml{GLB}$ \citep{Japaridze1988,Boolos1993}. We also show that even logics that are unsound with respect to such algebras do not have to be more complex than the classical propositional calculus. On the other hand, we observe that it is undecidable whether a syntactically defined logic is $\mathcal{V}$-complete.  After these results, we generalize the Blok Dichotomy \citep{Blok1978} to degrees of \cfV-incompleteness. In the end, we return to van Benthem's theme of syntactic aspects of modal incompleteness.\end{abstract} 
 
\tableofcontents
 
\section{Introduction}\label{sec:intro}
 
The discovery of \textit{Kripke incompleteness}, the existence of normal modal logics that are not sound and complete with respect to any class of Kripke frames, has been called one of the two forces that gave rise to the ``modern era'' of modal logic \citep[p.~44]{Blackburn2001}.  In the Lemmon Notes of 1966, it was conjectured that all normal modal logics are Kripke complete \citep[p.~74]{Lemmon1977}.\footnote{In fact, Wolter and  Zakharyaschev \citeyearpar[p.~428]{WolterZ06} claim that much of early research in the area was motivated by the still more optimistic \emph{Big Programme} or \emph{globalist's dream}. See \S\ \ref{subsec:RecV} and especially Footnote \ref{ft:globalist} for a further discussion.}  But this was not to be. Kripke incompleteness was first demonstrated with a bimodal logic \citep{Thomason1972}, shortly thereafter with complicated unimodal logics \citep{Fine1974b,Thomason1974}, and later with simple unimodal logics \citep{Benthem1978,Benthem1979,Boolos1985}. The significance of these discoveries can be viewed from several angles. 

From one angle, they show that Kripke frames are too blunt an instrument to characterize normal modal logics in general. More fine-grained semantic structures are needed. From another angle, they show that the notion of derivability in a normal modal logic with a set $\Sigma$ of axioms is too weak to capture the notion of Kripke frame consequence, in the sense where $\Sigma\vDash\varphi$ iff every Kripke frame that validates every $\sigma\in\Sigma$ also validates $\varphi$. A deep result of Thomason \citeyearpar{Thomason1975c} showed such weakness to be inevitable: the standard consequence relation for monadic second-order formulas with a single binary relation is reducible to the Kripke frame consequence relation for modal formulas.\footnote{To be more specific, Thomason \citeyearpar{Thomason1975c} showed that such second-order formulas can be translated into unimodal formulas in such a way that a second-order formula $\varphi$ is a consequence of a set $\Sigma$ of second-order formulas over standard second-order structures iff the translation $t(\varphi)$ of $\varphi$ is valid in every Kripke frame that validates $\{t(\sigma)\mid \sigma\in\Sigma\}\cup\{\delta\}$ for a fixed modal formula $\delta$.} Since the former consequence relation is not recursively axiomatizable, neither is the latter. 

Both of these angles on Kripke incompleteness lead to closely related questions and lines of investigation, with slightly differing focus. As we shall see, both of these lines of investigation will come together nicely in this paper. 

\subsection{The Semantic Angle}\label{sec:angle1}

The first angle on Kripke incompleteness---the realization that Kripke frames are not fine-grained enough for the study of normal modal logics in general---renewed interest in the \textit{algebraic} semantics for normal modal logics based on \textit{Boolean algebras with operators} (\baos{}) \citep{Jonsson1952a,Jonsson1952b}.  A \bao{} is a Boolean algebra together with one or more unary\footnote{J\'{o}nsson and Tarski considered operators of higher arity, but we will consider only unary operators.} \emph{operators}, i.e., unary operations $\Diamond$ such that for all elements $x,y$ of the algebra, $\Diamond (x\vee y)=\Diamond x\vee\Diamond y$, and for the bottom element $\bot$ of the algebra, $\Diamond\bot=\bot$. Every normal modal logic is sound and complete with respect to a \bao{}, namely, the Lindenbaum-Tarski algebra of the logic, according to a straightforward definition of when a modal formula is valid over a \bao{}. Kripke incompleteness can be better understood in light of the fact that Kripke frames correspond to \baos{} that are \textit{complete} ($\mathcal{C}$), \textit{atomic} ($\mathcal{A}$), and \textit{completely additive} ($\mathcal{V}$), or $\mathcal{CAV}$-\baos{} for short (see \S~\ref{sec:background}). A \bao{} is complete or atomic, respectively, according to whether its Boolean reduct is complete or atomic in the standard sense, while it is completely additive iff the following holds:
\begin{itemize}
\item[($\mathcal{V}$)] for any set $X$ of elements, if $\bigvee X$ exists, then $\bigvee\{\Diamond x\mid x\in X\}$ exists and \[\Diamond\bigvee X=\bigvee\{\Diamond x\mid x\in X\},\]
\end{itemize}
which in the case of a complete \bao{} reduces to the distributivity of $\Diamond$ over arbitrary joins. (The letter `$\mathcal{V}$' is intended to suggest the big join $\bigvee$.) 

Given the correspondence between Kripke frames and $\mathcal{CAV}$-\baos{}, the fact that a normal modal logic is not the logic of any class of Kripke frames means that it is not the logic of any class of $\mathcal{CAV}$-\baos{}. To put these points in more algebraic terms: normal modal logics correspond to varieties of \baos{}, and Kripke incompleteness is the phenomenon that not every variety of \baos{} can be generated as the smallest variety containing some class of $\mathcal{CAV}$-\baos{}. The last point is underscored by the dramatic \textit{Blok Dichotomy} \citep{Blok1978}: a variety of \baos{} is either uniquely generated by the $\mathcal{CAV}$-\baos{} it contains (the associated modal logic is strictly Kripke complete) or else there are continuum-many other varieties of \baos{} that contain exactly the same $\mathcal{CAV}$-\baos{} (other modal logics that are valid over exactly the same Kripke frames).

From this algebraic perspective, a series of natural questions arises: what happens if we drop or weaken one or more of the properties $\mathcal{C}$, $\mathcal{A}$, and $\mathcal{V}$? Do we thereby obtain distinct notions of completeness for normal modal logics? Can we represent the resulting \baos{} with some kind of frames with richer structure than Kripke frames? Does incompleteness persist even if we retain only one of the properties $\mathcal{C}$, $\mathcal{A}$, or $\mathcal{V}$? Does the analogue of the Blok Dichotomy hold if we drop one or more of  $\mathcal{C}$, $\mathcal{A}$, or $\mathcal{V}$?

These questions and the general landscape of sub-Kripkean notions of completeness from an algebraic perspective were the subject of the PhD research of Litak  \citeyearpar{Litak2004,Litak2005,Litak2005b,Litak2008}. As it turns out, none of the following notions of completeness are equivalent to any of the others: completeness with respect to atomic \baos{} ($\mathcal{A}$-\baos{}), complete \baos{} ($\mathcal{C}$-\baos{}), \baos{} that admit residuals ($\mathcal{T}$-\baos{}), atomic and completely additive \baos{} ($\mathcal{AV}$-\baos{}), complete and completely additive \baos{} ($\mathcal{CV}$-\baos{}), and $\mathcal{CAV}$-\baos{}. A rich hierarchy of different notions of  completeness  thereby comes into view. As for the representation question: Do\v{s}en \citeyearpar{Dosen1989} showed that $\mathcal{CA}$-\baos{} are the duals of \textit{normal neighborhood frames}; ten Cate and Litak \citeyearpar{tenCate2007} showed that $\mathcal{AV}$-\baos{} are the duals of \textit{discrete general frames}; and Holliday \citeyearpar{Holliday2015} has given a dual representation of $\mathcal{CV}$-\baos{} in the framework of \textit{possibility semantics}. As for the question of whether incompleteness persists if we retain only one of the properties $\mathcal{C}$, $\mathcal{A}$, or $\mathcal{V}$: Litak \citeyearpar{Litak2004} showed that there are $\mathcal{C}$-incomplete logics, and Venema \citeyearpar{Venema2003} showed that there are $\mathcal{A}$-incomplete logics (by proving the much stronger result that there is a variety of \baos{} all of whose members are atomless). Finally, as for extending the Blok Dichotomy: Zakharyaschev et al. \citeyearpar{Zak2001} noted that it extends to $\mathcal{T}$-\baos{}, Chagrova \citeyearpar{Chagrova1998} showed that it extends to $\mathcal{CA}$-\baos{}, and Litak \citeyearpar{Litak2008} showed that it extends to all $\mathcal{C}$-\baos{} (indeed, even to $\omega$-complete \baos{}, only requiring countable joins) as well as $\mathcal{AV}$-\baos{}. 

One piece of the puzzle remained missing for over a decade: are there $\mathcal{V}$-incomplete logics (\citealt{Litak2004}, \citealt[Ch.~9]{Litak2005b},  \citealt[\S~7]{Litak2008})? Venema \citeyearpar[\S~6.1]{Venema2007} also asked whether there are $\mathcal{V}$-inconsistent logics, i.e., normal modal logics that are not sound over any $\mathcal{V}$-\bao{}. In this paper, we answer these questions affirmatively. Our solution involves a first-order reformulation of the ostensibly second-order condition of complete additivity, which arose in the context of possibility semantics mentioned above \citep{Holliday2015}.\footnote{Around the same time that we proved that complete additivity of an operator in a \bao\ is a $\forall\exists\forall$ first-order condition, in January 2015, Hajnal Andr\'{e}ka, Zal\'{a}n Gyenis, and Istv\'{a}n N\'{e}meti independently proved that complete additivity of an operator on a poset is preserved under ultraproducts. After they learned of our result on the first-orderness of complete additivity in \baos\ from Steven Givant, Andr\'{e}ka et al. \citeyearpar{Andreka2016} extended it from \baos\ to arbitrary posets.} Using this first-order reformulation, we revisit an intriguing Kripke-incomplete logic of van Benthem \citeyearpar{Benthem1979} and show that van Benthem's logic, previously known to be $\mathcal{AV}$- and $\mathcal{T}$-incomplete, is the missing example of a $\mathcal{V}$-incomplete logic. Furthermore, it can be easily modified to answer Venema's question on $\mathcal{V}$-inconsistency. Building on this example, we extend the Blok Dichotomy to $\mathcal{V}$-incompleteness.

Given these results, the question arises of whether there are ``naturally occurring'' $\mathcal{V}$-incomplete logics. In the case of bimodal logic, we answer this question affirmatively. We show that the bimodal provability logic $\nml{GLB}$ \citep{Japaridze1988,Boolos1993}, well-known to be Kripke incomplete and hence $\mathcal{CAV}$-incomplete, is also $\mathcal{V}$-incomplete.

\subsection{The Syntactic Angle}\label{sec:angle2}

Van Benthem's logic was designed to illuminate the second angle on Kripke incompleteness mentioned above---the weakness of the notion of derivability in a normal modal logic with axioms---so our story will bring these two angles together. Let $\Sigma\vdash^{mnu}\varphi$ mean that $\varphi$ belongs to the smallest normal modal logic that contains the formulas in $\Sigma$ as axioms; so thinking in terms of derivations, not only modus ponens but also necessitation and uniform substitution may be applied to formulas from $\Sigma$. Van Benthem observed that Kripke incompleteness results can be viewed as \textit{non-conservativity} results with respect to $\vdash^{mnu}$. These results show that (i) for some modal formulas $\sigma$ and $\varphi$, $\{\sigma\}\nvdash^{mnu}\varphi$, yet (ii) every Kripke frame that validates $\sigma$ also validates $\varphi$. As is well known, every modal formula $\varphi$ can be translated into a sentence $SO(\varphi)$ of monadic second-order logic such that $\varphi$ is valid over a Kripke frame in the sense of Kripke semantics iff $SO(\varphi)$ is true in the frame as a standard second-order structure. Thus, (ii) can be rephrased as the fact that every Kripke frame that makes $SO(\sigma)$ true also makes $SO(\varphi)$ true. Van Benthem observed that the proof of (ii) typically shows that $SO(\varphi)$ is derivable from $SO(\sigma)$ using some weak system of monadic second-order logic plus an axiom of choice. In this sense, the second-order system is not conservative with respect to $\vdash^{mnu}$, in light of (i). To better gauge the weakness of $\vdash^{mnu}$, van Benthem asked whether there exist such a $\varphi$ and $\sigma$ for which the derivation of $SO(\varphi)$ from $SO(\sigma)$ can be carried out using only what he considered the weakest reasonable second-order system, dubbed \textit{weak second-order logic}. This would be a striking example of the weakness of $\vdash^{mnu}$ compared to second-order logic. Van Benthem indeed found such a $\varphi$ and $\sigma$. We will call the smallest normal modal logic containing that $\sigma$, which does not contain $\varphi$ by (i), the logic $\vblog$.

We will show that $\vblog$ is a $\mathcal{V}$-incomplete logic by showing that every $\mathcal{V}$-\bao{} that validates $\sigma$ also validates $\varphi$. At the end of the paper, we will follow a path in the spirit of \citealt{Benthem1979}: how can we strengthen our base logic to derive the formula $\varphi$ from $\sigma$, and what does this show about the weakness of the base logic? We begin by reviewing van Benthem's approach of translation into weak second-order logic. As it turns out, even weak second-order logic is much more powerful than what one needs to derive van Benthem's $\varphi$ from his $\sigma$ and thereby demonstrate non-conservativity with respect to basic modal logic. We consider two ways of extending the basic modal syntax for this purpose. One way leads to the \textit{nominal}\footnote{\label{f:firsthyb} The more recent use of the corresponding adjective in theoretical computer science \citep{Pitts2013,Pitts2016} has nothing to do with the term as used in \S\ \ref{subsec:hybrid} and references quoted therein.} calculus that characterizes consequence over $\mathcal{AV}$-\baos{}. Another way leads to the \textit{tense} calculus that characterizes consequence over $\mathcal{T}$-\baos{}. Finally, we find a common core for these weakenings: we show how the first-order reformulation of complete additivity inspires additional modal inference rules admissible over $\mathcal{V}$-\baos{} that allow us to derive van Benthem's $\varphi$ from his $\sigma$.  Of course, \emph{characterizing} a consequence relation by means of a (set of) rule(s), in a possibly extended syntax, is much more than just admissibility: in addition to soundness, one requires a generic completeness result. We will leave as an open question whether the rules we are proposing yield a syntactic characterization of $\mathcal{V}$-consequence. 

\subsection{Organization}

The paper is organized as follows. In \S~\ref{sec:VB}, we review the proof that van Benthem's logic $\vblog$ is Kripke incomplete, which we find to be a simple, vivid, and hence pedagogically useful example of Kripke incompleteness. In \S~\ref{sec:background}, we review the algebraic approach to modal (in)completeness as in \citealt{Litak2005b}. With this background, we proceed to the main part of the paper: in \S~\ref{sec:R&V}, we present the first-order reformulation of complete additivity, and in \S~\ref{sec:V-inc}, we use this reformulation to prove that the unimodal logic $\vblog$ and the bimodal logic $\nml{GLB}$ are $\mathcal{V}$-incomplete as above (and that the quasi-normal logic $\nml{GLSB}$ is even \cfV-inconsistent in a suitably adjusted sense).  In \S~\ref{sec:DecIssues}, we discuss issues of decidability and complexity: we show that even \cfV-\emph{inconsistent} logics do not have to be more complex than the classical propositional calculus, that the property of \cfV-completeness is in general undecidable, but that the associated notion of consequence (unlike Kripke frame consequence) is recursively axiomatizable.  In \S~\ref{sec:Blok}, we build on the example of $\vblog$ to generalize the Blok Dichotomy to $\mathcal{V}$-incompleteness. In the remaining sections we discuss the  second, syntactic angle presented above: in \S~\ref{sec:syntax}, we give syntactic proofs in existing derivation systems of the formula witnessing the incompleteness of $\vblog$; and in \S~\ref{sec:newsyntax}, we discuss extending our base logic with new $\mathcal{V}$-sound rules of inference. We conclude the paper in \S~\ref{sec:conclude} with open problems for future research.

\section{Kripke Incompleteness}\label{sec:VB}

In order to put Kripke incompleteness in context, let us review some basic definitions.

Let $\mathcal{L}$ be the set of formulas of propositional modal logic generated from a set $\{p_n\}_{n\in\mathbb{N}}$ of propositional variables. We use the usual notation for connectives: $\to$, $\wedge$, $\vee$, $\bot$,  $\Box$, $\Diamond$. A \textit{normal modal logic} is a set $\nml{L}\subseteq\mathcal{L}$ such that: (a) $\nml{L}$ contains all \textit{tautologies} of classical propositional logic; (b) $\nml{L}$ is closed under \textit{modus ponens}, i.e., if $\varphi\in\nml{L}$ and $\varphi\rightarrow\psi\in\nml{L}$, then $\psi\in\nml{L}$; (c) $\nml{L}$ is closed under \textit{uniform substitution}, i.e., if $\varphi\in\nml{L}$ and $\psi$ is the result of uniformly substituting formulas for propositional variables in $\varphi$, then $\psi\in\nml{L}$; (d) $\nml{L}$ is closed under \textit{necessitation}, i.e., if $\varphi\in\nml{L}$, then $\Box\varphi\in\nml{L}$; and (e) $\Box(p_0\rightarrow p_1)\rightarrow (\Box p_0\rightarrow \Box p_1)\in\nml{L}$. Let $\nml{K}$ be the smallest normal modal logic.

The definition of a normal modal logic extends to polymodal languages with multiple modalities $\Box_1$, $\Box_2$, etc., by requiring (d) and (e) for each $\Box_i$. In this section, we focus on the unimodal language, but polymodal languages will become important in \S~\ref{sec:V-inc}.

We assume familiarity with Kripke frames $\ffr=\langle W,R\rangle$, Kripke models $\mathcal{M}=\langle W,R,V\rangle$, and the Kripke semantic definition of when a formula is true at a $w\in W$: $\mathcal{M},w\vDash\varphi$. We abuse notation and write `$w\in\ffr$' or `$w\in\mathcal{M}$' to mean $w\in W$. For convenient additional notation, given $w\in W$, let $R(w)=\{v\in W\mid wRv\}$, and given $\varphi\in\mathcal{L}$, let $\llbracket \varphi\rrbracket^\mathcal{M}=V(\varphi)=\{v\in W\mid \mathcal{M},v\vDash\varphi\}$. A formula $\varphi$ is \textit{globally true} in a Kripke model $\mathcal{M}$, written `$\mathcal{M}\vDash\varphi$', iff $\mathcal{M},w\vDash\varphi$ for every $w\in\mathcal{M}$; and $\varphi$ is \textit{valid} over a Kripke frame $\ffr=\langle W,R\rangle$, written `$\ffr\vDash\varphi$', iff $\mathcal{M}\vDash\varphi$ for every model $\mathcal{M}=\langle W,R,V\rangle$ based on $\ffr$. A formula is valid over a class $\mathsf{F}$ of Kripke frames iff it is valid over every frame in the class. Let $\mathrm{Log}(\mathsf{F})$ be the set of formulas valid over $\mathsf{F}$, which is always a normal modal~logic.

A logic $\nml{L}$ is \textit{Kripke complete} if there is a class $\mathsf{F}$ of Kripke frames for which $\nml{L}=\mathrm{Log}(\mathsf{F})$. Otherwise it is \textit{Kripke incomplete}. For any logic $\nml{L}$, we can consider the class of frames that validate it: $\mathrm{Fr}(\nml{L})=\{\ffr\mid \ffr\vDash \varphi\mbox{ for all }\varphi\in\nml{L}\}$. For a Kripke complete logic, $\nml{L}=\mathrm{Log}(\mathrm{Fr}(\nml{L}))$, whereas for a Kripke incomplete logic, $\nml{L}\subsetneq\mathrm{Log}(\mathrm{Fr}(\nml{L}))$.

Everything said above was in terms of \textit{validity}, but we could also put our discussion in terms of \textit{consequence}. To avoid confusion, it is important to distinguish between the following consequence relations, following van Benthem \citeyearpar[p.~37]{Benthem1983}:
\begin{itemize}
\item $\Sigma\vDash_{M,l}\varphi$ iff for every Kripke model $\mathcal{M}$ and $w\in \mathcal{M}$, if $\mathcal{M},w\vDash\sigma$ for all $\sigma\in\Sigma$, then $\mathcal{M},w\vDash\varphi$ (local consequence over models).
\item $\Sigma\fcons\varphi$ iff for every Kripke frame $\ffr$, if $\ffr\vDash \sigma$ for all $\sigma\in\Sigma$, then $\ffr\vDash \varphi$ (global consequence over frames).
\end{itemize}
These two notions of consequence are related to two different notions of when $\varphi$ is derivable from a set $\Sigma$ of formulas. Let $\Sigma\vdash_\nml{K}^m\varphi$ iff $\varphi$ belongs to the closure of $\nml{K}\cup\Sigma$ under modus ponens. Let $\Sigma\vdash_\nml{K}^{mnu}\varphi$ iff $\varphi$ belongs to the smallest normal modal logic $\nml{L}\supseteq\Sigma$, which is the closure of $\nml{K}\cup\Sigma$ under modus ponens, necessitation, and uniform substitution. Our $\vdash_\nml{K}^{mnu}$ is what van Benthem \citeyearpar{Benthem1979} denotes by `$\vdash_K$' and calls `the minimal modal logic $K$'. It would be reasonable to call $\vdash_\nml{K}^{mnu}$ \textit{derivability from axioms} and to call $\vdash_\nml{K}^m$ \textit{derivability from premises}, since intuitively necessitation and uniform substitution should be applicable to logical axioms but not to arbitrary premises.

The relation $\vDash_{M,l}$ is axiomatized by $\vdash^{m}_\nml{K}$: $\Sigma\vDash_{M,l}\varphi$ iff $\Sigma\vdash_\nml{K}^m\varphi$, which is equivalent to there being $\sigma_1,\dots,\sigma_n\in\Sigma$ such that $(\sigma_1\wedge\dots\wedge\sigma_n)\rightarrow\varphi\in\nml{K}$. By contrast, by the result of Thomason \citeyearpar{Thomason1975c} mentioned in \S~\ref{sec:intro}, $\fcons$ is not recursively axiomatizable. We have that $\Sigma\vdash_\nml{K}^{mnu}\varphi$ implies $\Sigma\fcons\varphi$, but the converse is not guaranteed. The ``weakness'' of $\vdash^{mnu}$ referred to in \S~\ref{sec:angle2} is the fact that $\Sigma\fcons\varphi$ does not guarantee $\Sigma\vdash_\nml{K}^{mnu}\varphi$.

While $\vDash_{M,l}$ seems to capture an intuitive notion of modal consequence,  the relation $\fcons$ implicitly prefixes all premises by arbitrary sequences of boxes and  universal quantifiers over propositional variables, so it yields striking consequences like $\{p\}\fcons\bot$ and $\{\varphi\}~\fcons~\Box\varphi$. Thus, it does not enjoy a  deduction theorem and one should not think about sets of formulas closed under $\fcons$ as ``local theories''. There is, however, a better way to think about $\fcons$. In modal logic, we are often not interested in the class of all frames, but rather in some restricted classes of frames, perhaps defined as $\mathrm{Fr}(\Sigma)$ for some set $\Sigma\subseteq\mathcal{L}$. We then want to know what formulas are valid over this class, i.e., whether $\varphi\in\mathrm{Log}(\mathrm{Fr}(\Sigma))$. This is equivalent to asking whether $\Sigma\fcons\varphi$. Also note that $\nml{L}$ is Kripke complete as above iff for every $\varphi\in\mathcal{L}$, $\nml{L}\fcons\varphi$ implies $\varphi\in\nml{L}$.

Let $\vblog$ be the smallest normal modal logic containing the axiom 
\[ \Box\Diamond\top\rightarrow \Box (\Box (\Box p\rightarrow p)\rightarrow p),\]
which we will call the $\vblog$-axiom. Van Benthem \citeyearpar{Benthem1979} proved that the logic $\vblog$ is Kripke incomplete. While $\vblog$ may at first seem  an entirely ad hoc example of a Kripke-incomplete logic, we will observe a striking connection between the incompleteness of $\vblog$ and the incompleteness of an important provability logic in \S~\ref{subsec:GLB}. In this connection, it is noteworthy that the $\vblog$-axiom is a theorem of the provability logic $\nml{GL}$, the smallest normal modal logic containing the L\"{o}b axiom, $\Box (\Box p\to p)\to \Box p$. Substituting $\bot$ for $p$ in the L\"{o}b axiom yields $\Box\Diamond \top \to \Box\bot$, which in the context of provability logic is a modal version of G\"{o}del's Second Incompleteness Theorem.\footnote{\label{footnote:vbsource}When describing the origin of the $\nml{vB}$ axiom, van Benthem \citeyearpar{Benthem1979} hints at a \emph{quasi-normal} logic that one recognizes as (a subsystem of) Solovay's system $\nml{GLS}$ (see, e.g., \citealt[p.~65]{Boolos1993}), another central formalism in the area of provability logic. The formula under $\Box$ in the consequent of  $\nml{vB}$, i.e., $\Box (\Box p\rightarrow p)\rightarrow p$, is a theorem of $\nml{GLS}$, and as shown by a syntactic derivation in \citealt{Benthem1979}, this formula alone makes it impossible to characterize $\nml{GLS}$ in terms of Kripke semantics with distinguished worlds, which is the standard relational semantics for quasi-normal systems. We will discuss such provability-related quasi-normal systems in \S~\ref{subsec:nonnormal}.
} Clearly the $\vblog$-axiom is derivable from $\Box\Diamond \top \to \Box\bot$. In the other direction, van Benthem showed that $\Box\Diamond \top \to \Box\bot$ is a Kripke-frame consequence of the $\vblog$-axiom. However, he also showed that $\Box\Diamond \top \to \Box\bot$ is \textit{not} a theorem of $\vblog$. Together these facts imply  the Kripke-incompleteness of $\vblog$.

So that our presentation is somewhat self-contained, we will include van Benthem's proof of the Kripke-incompleteness of $\vblog$. It should be emphasized that this is one of the simplest proofs of Kripke incompleteness.\footnote{It is not, however, the most natural example of a Kripke incomplete unimodal logic. That honor goes to the logic of the Henkin sentence $\Box(\Box p\leftrightarrow p)\rightarrow \Box p$ (see \citealt{Boolos1985}), which is also a simplest possible Kripke incomplete unimodal logic in the following sense: it is axiomatized by a formula with only one propositional variable and modal depth 2. Lewis \citeyearpar{Lewis1974} showed that all normal unimodal logics axiomatizable by formulas of modal depth $\leq 1$ are Kripke complete. Despite its charms, the Henkin logic is irrelevant for our purposes in this paper, for a reason that can be explained using notions introduced in \S~\ref{sec:background}: the proof of its Kripke incompleteness does not attack complete additivity, but rather closure under countable joins/meets, i.e., $\clofr{\omega C}$-completeness. This property is exploited by most incompleteness proofs involving $\nml{GL}$ and its relatives, like the failure of strong completeness of $\nml{GL}$ itself or the Kripke-inconsistency of various tense logics containing $\nml{GL}$ (although counterexamples related to the original one by \citealt{Thomason1972} clash with full \cfC\ rather than its restriction to $\clofr{\kappa C}$ for any fixed cardinality $\kappa$) \citep{Litak2005b}.} It should also be emphasized that in \S~\ref{sec:V-inc} we will prove a much more general result than the following lemma; but we include a proof of Lemma \ref{lem:1} for later reference in \S~\ref{sec:syntax}.

\begin{lemma}\label{lem:1} Any Kripke frame that validates $\vblog$ also validates $\Box\Diamond\top\to\Box\bot$.
\end{lemma}
\begin{proof} Let $\ffr$ be a Kripke frame that validates $\vblog$. We need to show that if $x\in\ffr$ is such that $R(x)\not=\varnothing$, then there is a $w\in R(x)$ such that $R(w)=\varnothing$ (think of the contrapositive: $\Diamond\top\to\Diamond\Box\bot$). Thus, consider an $x\in\ffr$ such that $R(x)\not=\varnothing$ and a $y\in R(x)$. Let $\mathcal{M}$ be a model based on $\ffr$ with $\llbracket p\rrbracket^\mathcal{M}=\{u\in \ffr\mid u\not = y\}$. For reductio, suppose there is no $w\in R(x)$ with $R(w)=\varnothing$, so $\mathcal{M},x\vDash \Box\Diamond\top$. Then since $\ffr$ validates $\vblog$, we have $\mathcal{M},x\vDash \Box (\Box (\Box p\rightarrow p)\rightarrow p)$, which with $y\in R(x)$ implies $\mathcal{M},y\vDash \Box(\Box p \rightarrow p)\rightarrow p$, which with our valuation for $p$ implies $\mathcal{M},y\nvDash \Box(\Box p \rightarrow p)$. Thus, there is a $z\in R(y)$ such that $\mathcal{M},z\vDash \Box p$ but $\mathcal{M},z\nvDash p$. From $\mathcal{M},z\nvDash p$, we have $z=y$.  Then from $z\in R(y)$ and $\mathcal{M},z\vDash \Box p$, we have $y\in R(y)$ and $\mathcal{M},y\vDash \Box p$, so $\mathcal{M},y\vDash p$, a contradiction. Hence there is a $w\in R(x)$ with $R(w)=\varnothing$, as desired.
\end{proof}

All that remains to show is that $\Box\Diamond\top\to\Box\bot\not\in\vblog$. We can do so by exhibiting a Kripke model $\mathcal{M}$ and showing that  every $\varphi\in\vblog$ is globally true in $\mathcal{M}$, while $\Box\Diamond\top\to\Box\bot$ is not. Since every $\varphi\in\nml{K}$ is globally true in every Kripke model, and the set of formulas that are globally true in a given Kripke model $\mathcal{M}$ is closed under modus ponens and necessitation, to show that every $\varphi\in\vblog$ is globally true in $\mathcal{M}$, it suffices to show that every substitution instance of the $\vblog$-axiom is globally true in $\mathcal{M}$ (in the terminology of \citealt{Fine1974b}, the $\vblog$-axiom is \textit{strongly verified} in $\mathcal{M}$).

In the literature on Kripke incompleteness, rather than directly exhibiting a model $\mathcal{M}$ as above, authors typically exhibit an appropriate \textit{general frame}\footnote{Readers familiar with the algebraic approach to modal logic will of course note that one can directly define a modal algebra (\bao) instead. We will introduce modal algebras soon in \S~\ref{sec:background}.}  $\gfr=\langle W,R,\admis{W}\rangle$ where $\langle W,R\rangle$ is a Kripke frame and $\admis{W}$ is a family of subsets of $W$ that is closed under union, complement relative to $W$, and the operation $X\mapsto R^{-1}[X]=\{w\in W\mid \exists x\in X\colon wRx\}$. An \textit{admissible} model based on a general frame $\gfr=\langle W,R,\admis{W}\rangle$ is a model $\mathcal{M}=\langle W,R,V\rangle$ such that $\llbracket p\rrbracket^\mathcal{M}\in\admis{W}$ for every propositional variable $p$.  An easy induction then shows that for every $\varphi\in\mathcal{L}$, $\llbracket \varphi\rrbracket^\mathcal{M}\in\admis{W}$. It follows that if a formula $\psi$ is globally true in every admissible model based on $\gfr$---in which case $\psi$ is \textit{valid} over $\gfr$, written `$\gfr\vDash\psi$'---then for any particular admissible model $\mathcal{M}$ based on $\gfr$, all substitution instances of $\psi$ are globally true in $\mathcal{M}$. Thus, to obtain a model $\mathcal{M}$ as in the previous paragraph, it suffices to exhibit a general frame $\gfr$ over which the $\vblog$-axiom is valid, while $\Box\Diamond\top\to\Box\bot$ is not.

Another way to motivate going to a general frame here is by the following observation. Define a consequence relation $\vDash_G$ by $\Gamma\vDash_G\varphi$ iff for every general frame $\gfr$, if $\gfr\vDash\sigma$ for all $\sigma\in\Sigma$, then $\gfr\vDash\varphi$. Then it can be shown that there is an exact match between $\vDash_G$ and the derivability relation $\vdash_\nml{K}^{mnu}$ above: $\Sigma\vDash_G \varphi$ iff $\varphi$ belongs to the smallest normal modal logic containing $\Sigma$. So to show that $\Box\Diamond\top\to\Box\bot$ does not belong to $\vblog$, we simply show that $\{\vblog\mbox{-axiom}\}\nvDash_G \Box\Diamond\top\to\Box\bot$, which is again to show that there is a general frame $\gfr$ over which the $\vblog$-axiom is valid, while $\Box\Diamond\top\to\Box\bot$ is not.

\begin{definition}[van Benthem frame]\label{def:VBframe} The \textit{van Benthem frame} (see Figure \ref{vBfig}) is the general frame $\mathscr{VB}=\langle W, R, \admis{W}\rangle$ where:
\begin{enumerate}
\item $W=\mathbb{N}\cup\{\infty,\infty+1\}$;
\item $R=\{\langle \infty+1,\infty\rangle, \langle \infty,\infty\rangle\} \cup \{\langle \infty,n\rangle\mid n\in\mathbb{N}\}\cup \{\langle m,n\rangle\mid m,n\in\mathbb{N}, m>n\}$;\footnote{Cresswell \citeyearpar{Cresswell84} notes that we can leave out $\langle \infty,\infty\rangle$, but the $\infty$-reflexive variant will be more convenient in \S~\ref{subsec:Decide}.}
\item $\admis{W}=\{X\subseteq W\mid X\mbox{ is finite and }\infty\not\in X\}\cup \{X\subseteq W\mid X\mbox{ is cofinite and }\infty\in X\}$.
\end{enumerate}
\end{definition} 

Observe that $\admis{W}$ is closed under union, relative complement, and $X\mapsto R^{-1}[X]$.

 \begin{figure}[h]
\begin{center}
\begin{tikzpicture}[->,>=latex,shorten >=1pt,shorten <=1pt, auto,node
distance=2cm,every loop/.style={<-,shorten <=1pt}]
\tikzstyle{every state}=[fill=gray!20,draw=none,text=black]

\node (infty+1) at (0,0) {{$\infty+1$}};
\node (infty) at (2,0) {{$\infty$}};
\node at (3,0) {{$\dots$}};
\node (2) at (4,0) {{$2$}};
\node (1) at (6,0) {{$1$}};
\node (0) at (8,0) {{$0$}};

\path (infty+1) edge[->] node {{}} (infty);
\path (infty) edge[in=116,out=65,loop] node {{}} (infty);
\path (infty) edge[bend left,->] node {{}} (2);
\path (infty) edge[bend left,->] node {{}} (1);
\path (infty) edge[bend left,->] node {{}} (0);
\path (2) edge[->] node {{}} (1);
\path (2) edge[bend right,->] node {{}} (0);
\path (1) edge[->] node {{}} (0);

\end{tikzpicture}
\end{center}
\caption{The van Benthem frame \vbfr.}\label{vBfig}
\end{figure}

We now add the final piece of the argument.

\begin{lemma}\label{lem:2} $\Box\Diamond\top\rightarrow \Box (\Box (\Box p\rightarrow p)\rightarrow p)$ is valid over $\vbfr$, while  $\Box\Diamond\top\to\Box\bot$ is not. Thus, $\Box\Diamond\top\to\Box\bot\not\in\vblog$.
\end{lemma}

\begin{proof} Consider any admissible model $\mathcal{M}$ based on $\vbfr$. First observe that $0\in\llbracket \Box \bot\rrbracket^\mathcal{M}$, and for all $w\in W\setminus\{\infty+1\}$, $wR0$, so $w\not\in\llbracket \Box\Diamond\top\rrbracket^\mathcal{M}$; but $\infty+1\in\llbracket \Box\Diamond\top\rrbracket^\mathcal{M}$, so we have $\llbracket \Box\Diamond\top\rrbracket^\mathcal{M}=\{\infty+1\}$. Thus, we need only show that $\infty+1 \in\llbracket \Box (\Box (\Box p\rightarrow p)\rightarrow p)\rrbracket^\mathcal{M}$, which is equivalent to $\infty\in \llbracket \Box(\Box p\rightarrow p)\rightarrow p\rrbracket^\mathcal{M}$. If $\infty\in \llbracket \Box(\Box p\rightarrow p)\rrbracket^\mathcal{M}$, then for every $n\in\mathbb{N}$, $n\in\llbracket \Box p\rightarrow p\rrbracket^\mathcal{M}$, whence an obvious induction shows that $\mathbb{N}\subseteq\llbracket p\rrbracket^\mathcal{M}$. Hence $\llbracket p\rrbracket^\mathcal{M}$ is \textit{cofinite}, so $\infty\in \llbracket p\rrbracket^\mathcal{M}$. This shows that $\infty\in \llbracket \Box(\Box p\rightarrow p)\rightarrow p\rrbracket^\mathcal{M}$, which completes the proof that $\Box\Diamond\top\rightarrow \Box (\Box (\Box p\rightarrow p)\rightarrow p)$ is valid over $\vbfr$.

Finally, observe that $\infty+1 \not\in\llbracket \Box\Diamond\top\to\Box\bot\rrbracket^\mathcal{M}$.\end{proof}

Putting together Lemmas \ref{lem:1} and \ref{lem:2}, we have the claimed result.

\begin{theo}[\citealt{Benthem1979}] The logic $\vblog$ is Kripke incomplete.
\end{theo}
\noindent Van Benthem's main point was not that $\vblog$ is Kripke incomplete\footnote{Indeed, the very existence of Kripke incomplete logics was not much of a revelation anymore at the time. We have already mentioned in the opening paragraph of this paper that Kripke incompleteness was first demonstrated with a bimodal logic \citep{Thomason1972} and shortly thereafter with complicated unimodal logics \citep{Fine1974b,Thomason1974} located in increasingly specific areas of the lattice of extensions of $\nml{K}$. Van Benthem himself devoted an earlier paper \citeyearpar{Benthem1978} to simple examples of incomplete logics. Moreover, at that time two crucial results (which are going to be our main concern in \S~\ref{subsec:RecV} and \S\S~\ref{sec:Blok}-\ref{sec:syntax}) that make explicit the ubiquity and unavoidability of Kripke incompleteness were already known:  \citealt{Thomason1975c} and  \citealt{Blok1978}.} but that it is a special example of such incompleteness: it can be used to show that the derivability relation $\vdash_\nml{K}^{mnu}$ falls short of capturing not only the consequence relation $\fcons$ itself, but also syntactically inspired weakenings of $\fcons$. We will return to this in \S~\ref{sec:syntax}. 

The main point we wish to make about $\vblog$ is that it is special in another way: it provides the long missing example of a \textit{$\mathcal{V}$-incomplete} logic. To explain what this means and its context, let us now review the algebraic perspective on modal logic. 

\section{The Algebraic Approach to Modal (In)completeness}\label{sec:background}

As noted in \S~\ref{sec:intro}, the discovery of Kripke incompleteness renewed interest in the algebraic semantics for normal modal logics based on \emph{Boolean algebras with operators} (\baos)  \citep{Jonsson1952a,Jonsson1952b}. An $n$-ary \textit{operator} on a Boolean algebra with universe $A$ is a function $f\colon A^n\to A$ that preserves finite joins in each coordinate (including the join of the empty set, $\bot$); a \textit{dual operator} preserves finite meets in each coordinate (including the meet of the empty set, $\top$). A \bao\ is a Boolean algebra equipped with a collection of operators. In this paper, we consider only \baos\ with unary operators. If the collection of these operators in a \bao\ $\Ag$ has cardinality $\kappa$, we call $\mathfrak{A}$ a $\kappa$-\bao. Per tradition, we call a $1$-\bao\ a \emph{modal algebra} (\ma).

The language of basic unimodal logic can be interpreted in an \ma\ $\Ag$ in the obvious way: any mapping $\theta$ of propositional variables to elements of $\Ag$ extends to a mapping $\hat{\theta}$ of arbitrary formulas to elements of $\Ag$, taking $\hat{\theta}(p)=\theta(p)$, $\hat{\theta}(\neg\varphi)= -\hat{\theta}(\varphi)$, $\hat{\theta} (\varphi\vee\psi)=\hat{\theta}(\varphi)+\hat{\theta}(\psi)$, and $\hat{\theta}(\Diamond\varphi)=f(\hat{\theta}(\varphi))$, where $-$, $+$, and $f$ are the complement, join, and operator in $\Ag$, respectively. The \ma\ $\Ag$ \textit{validates} a modal formula $\varphi$ (notation: $\Ag \vDash \varphi$) iff every such mapping sends $\varphi$ to the top element $\top$ of $\Ag$. For a given class $\mathcal{X}$ of \mas\ (see below for important examples of such classes), we define a consequence relation $\vDash_\mathcal{X}$, analogous to the global Kripke frame consequence relation $\fcons$ from \S~\ref{sec:VB}:
\begin{itemize}
\item $\Sigma\vDash_\mathcal{X}\varphi$ iff for every $\mathfrak{A}\in \mathcal{X}$, if $\mathfrak{A}\vDash \sigma$ for all $\sigma\in\Sigma$, then $\mathfrak{A}\vDash\varphi$.
\end{itemize}

All of the notions above extend to interpreting a polymodal language with $\kappa$ modal operators in $\kappa$-\baos\ in the obvious way.

\begin{definition}\label{def:complete} Let $\mathcal{X}$ be a class of $\kappa$-\baos\ and $\nml{L}$ a normal modal logic in a language with $\kappa$ modal operators. We say that $\nml{L}$ is \textit{$\mathcal{X}$-complete} if for all formulas $\varphi$, we have $\varphi\in\nml{L}$ iff $\nml{L}\vDash_\mathcal{X}\varphi$.  Otherwise $\nml{L}$ is \textit{$\mathcal{X}$-incomplete}.
\end{definition}
Equivalently, $\nml{L}$ is $\mathcal{X}$-complete if $\nml{L}$ is the logic of some class $\mathcal{K}\subseteq\mathcal{X}$, i.e., $\nml{L}$ is exactly the set of formulas validated by all \baos\ in $\mathcal{K}$.

Each normal modal logic $\nml{L}$ is the logic of a \bao: the Lindenbaum-Tarski algebra of $\nml{L}$, whose elements are the equivalence classes of modal formulas under the relation defined by $\varphi \sim\psi$ iff $\varphi\leftrightarrow\psi\in\nml{L}$, and whose operations are defined in the obvious way: $-[\varphi]=[\neg\varphi]$, $[\varphi]+[\psi]=[\varphi\vee\psi]$, and $f_i([\varphi])=[\Diamond_i\varphi]$. This general algebraic completeness theorem via Lindenbaum-Tarski algebras can be seen as a special case of an even more general approach: since the derivability relation\footnote{The relation $\vdash_\nml{L}^{mn}$ is defined by: $\Sigma \vdash_\nml{L}^{mn} \varphi$ iff $\varphi$ belongs to the closure of $\mathsf{L}\cup \Sigma$ under modus ponens and necessitation. For $\nml{L}=\nml{K}$, $\vdash_\nml{K}^{mn}$ is the derivability relation that matches \textit{global} consequence over Kripke \textit{models}: $\Sigma\vDash_{M}\varphi$ iff for every Kripke model $\mathcal{M}$, if $\mathcal{M}\vDash \sigma$ for all $\sigma\in\Sigma$, then $\mathcal{M}\vDash \varphi$.} $\vdash_\nml{L}^{mn}$ associated with a given normal modal logic $\nml{L}$ is (\emph{Rasiowa}) \emph{implicative} and hence \emph{strongly finitely algebraizable}, one obtains a strong completeness theorem for $\nml{L}$ with respect to algebraic semantics using the standard machinery of abstract algebraic logic (AAL)   \citep{Rasiowa74:aatnl,BlokP89:ams,Czelakowski2001,Andreka2001,Font06:sl,FontJP03a:sl,FontJP09:sl}; see \S~\ref{subsec:nonnormal} and especially Footnote \ref{footnote:alghist} for historical origins of this approach.

Before proceeding further, let us fix notation for dealing with algebras. We use gothic letters $\Ag, \Bg, \Cg \dots$ for names of algebras and $a, b, c \dots$ for elements of algebras. Whenever it is not confusing, we blur the distinction between an algebra and its carrier, writing statements like `$a \in \Ag$'.  We also  blur the distinction between modal formulas and \bao-terms, and henceforth we will simply use $\neg$, $\vee$, and $\wedge$ for the complement, join, and meet, respectively, in our algebras, trusting that no confusion will arise. In an \ma, we take $\Diamond$ to be the operator and $\Box$ to be a dual operator, defined by $\Box a =\neg\Diamond\neg a$. In \baos, we may add indices to distinguish between multiple operators, e.g., taking $\langle 0\rangle$ and $\langle 1\rangle$ to be operators and $[0]$ and $[1]$ to be their duals.

The generic completeness result described above made the algebraic semantics historically the first to be studied, prior to the invention of Kripke frames (see \citealt[\S~3]{Goldblatt03:jal}; \citealt[\S~1.7]{Blackburn2001}).  However, one can also obtain a generic completeness result with respect to the general frames of \S~\ref{sec:VB}. This result is implicit already in the work of J\'onsson and Tarski  
\citeyearpar{Jonsson1952a,Jonsson1952b}, who proposed an extension of Stone's Representation Theorem from Boolean algebras to \baos. The general frames obtained via this representation are known as \emph{descriptive} frames; thus, every normal modal logic is sound and complete with respect to a class of descriptive frames.  Furthermore, for a large class of modal axioms/equations, especially so-called Sahlqvist axioms and their various generalizations (see, e.g., \citealt{ConradieGV06} or \citealt{ConradieGP14} and references therein), one can in addition observe their \emph{persistence} in passing from a descriptive frame to its underlying Kripke frame. This phenomenon is known in the contemporary literature as \emph{canonicity}  or \emph{d-persistence} (\citealt[Ch.~10]{Chagrov1997}, \citealt[Ch.~5]{Blackburn2001}), but in hindsight the J\'onsson-Tarski work can be seen as its earliest study. In short, algebra combined with duality theory provides a viable route towards Kripke completeness results for suitably well-behaved logics. On the other hand, as far as \emph{weak} completeness (which is the main subject of the present paper) with respect to \emph{finite} models is concerned,  it is not necessary to phrase such completeness results in algebraic terms or to involve the Stone-J\'onsson-Tarski duality in  the proof. Think, e.g., of tableaux-style extraction of countermodels from failed proof search in suitable Gentzen-style calculi (in fact, close to Kripke's original work) or the technique of \emph{normal forms}. Such finitary approaches are  not restricted to d-persistent logics.  See \citealt{Fine1975}, \citealt{Moss07}, and \citealt{BezhanishviliG14}; the relationship with duality theory and construction of free algebras 
  is discussed in  \citealt{Ghilardi1995}, \citealt{Bezhanishvili2007}, and \citealt{CoumansVG13}.

\begin{remark}
Henkin-style general-frame strong completeness  of logics \emph{in countably many variables} is equivalent to the weak K\"{o}nig Lemma even over $\nml{RCA}_0$, the weak base theory for reverse mathematics \cite[IV.3.3]{Simpson2009}. Recall that  the weak K\"{o}nig Lemma holds in $\nml{ZF}$ or even in the Zermelo set theory (i.e., $\nml{ZF}$ without replacement). On the other hand, for uncountably many propositional variables, one needs 
representation theorems in the style of J\'onsson and Tarski. They rely on the Ultrafilter Theorem, or equivalently, the Boolean Prime Ideal Theorem ($\nml{BPI}$). It would thus seem that completeness based on canonicity is rather non-constructive, but with some care it is possible to prove more fine-grained results along these lines---see \citealt{GhilardiM:97} and  \citealt{Suzuki2010} or strong completeness with respect to the general \textit{possibility frames} of \citealt{Holliday2015}. It is also worth mentioning that strong Kripke completeness does not imply d-persistence, as in the case of the tense logic of the reals \citep{Wolter96:properties}; if one is willing to extend the notions of strong completeness and canonicity to neighborhood frames, another counterexample is provided by the McKinsey logic \citep{Surendonk01}.
\end{remark}
 
In the reverse direction, we of course do not need Kripke frames (or any other semantics) as an intermediate step in proving algebraic completeness; the Lindenbaum-Tarski construction provides a direct route. Nevertheless, Kripke completeness results can be reformulated and understood from an algebraic point of view: they establish that the equational class (\emph{variety}) of \baos\ corresponding to a given modal logic is determined by its elements with special additional properties. In other words, they show that when looking for algebraic models refuting a given formula/equation, one can restrict attention to a well-behaved subclass of algebras. 
 
Let us discuss this in more detail. Recall the standard construction associating with a given Kripke frame $\ffr = \langle W, R \rangle$ its dual \ma\ $\ffr^+$, whose universe is  $\wp(W)$, whose Boolean operations are interpreted using the set-theoretic ones, and whose operator is defined by $\Diamond X= R^{-1}[X]$. This is just a special case of taking the dual of a general frame $\gfr=\langle W,R,\admis{W}\rangle$, where the \ma\ in question is provided by $\admis{W}$; in the case of Kripke frames, $\admis{W}=\wp(W)$. As observed already by J\'onsson and Tarski 
\citeyearpar{Jonsson1952a,Jonsson1952b}, such an \ma\ always has the following three special (and mutually independent) properties.

\begin{itemize}
\item[($\mathcal{C}$)]\textit{lattice-completeness}: given any set $X$ of elements of $\Ag$, its \emph{join} $\bigvee X$ exists in $\Ag$. This also implies the existence of arbitrary meets. 
\item[($\mathcal{A}$)]\textit{atomicity}:  any non-bottom element is above an \emph{atom}, i.e., a minimal non-bottom element.
\item[($\mathcal{V}$)] \textit{complete additivity}: for any set $X$ of elements, if $\bigvee X$ exists, then $\bigvee\{\Diamond x\mid x\in X\}$ exists and \[\Diamond\bigvee X=\bigvee\{\Diamond x\mid x\in X\}.\]
\end{itemize}
For complete \mas, complete additivity reduces to distributivity of $\Diamond$ over arbitrary joins. Of course, $\mathcal{V}$ can be equivalently stated with $\bigwedge$ replacing $\bigvee$ and $\Box$ replacing $\Diamond$.

\begin{remark}
In the case of (duals of) Kripke frames, where every subset is admissible, $\bigvee$ is simply $\bigcup$. Nevertheless---and this is an important point!---it does not have to be the case with general frames whose $\admis{W}$ is lattice-complete. In particular, in any descriptive frame associated with an infinite \cfC-\bao, there are instances of joins and meets not coinciding with unions and intersections. The atoms of (duals of) Kripke frames are obviously singleton sets $\{x\}$. Again, for arbitrary general frames this does not have to be the case, but this is less crucial: in \emph{differentiated} general frames, which are in an important sense the only relevant ones, admissible atoms have to be singleton sets.
\end{remark}

\begin{remark}\label{rem:V_n} As we shift from unimodal to polymodal contexts, $\mathcal{C}$ will be the class of complete \baos\ with the appropriate number of operators in the context, and similarly for $\mathcal{A}$, $\mathcal{V}$, etc. In principle, when we say a polymodal logic is \textit{$\mathcal{V}$-complete}, we should mean that it is complete with respect to \baos\ in which \textit{every operator} is completely additive. But most of the time, complete additivity of \emph{all} operators occurring is not needed, and we may wish to be more fine-grained: given a logic $\nml{L}$ with modalities indexed by natural numbers, we can say that $\nml{L}$ is \textit{$\mathcal{V}_n$-complete} if it is $\mathcal{X}$-complete, as in Definition \ref{def:complete}, where $\mathcal{X}$ is the class of  \baos\ in which the $n$-th operator is completely additive. 
\end{remark}
 
As it turns out, the combination of the three properties above is a defining feature of duals of Kripke frames. One can say even more: the category of Kripke frames with bounded morphisms is dually equivalent to that of \cfC\cfA\cfV-\baos\ with \emph{complete} morphisms \citep{Thomason75:categories}. In particular, taking any Kripke frame/\cfC\cfA\cfV-\bao, converting it into its dual \cfC\cfA\cfV-\bao/Kripke frame, and then going back produces an output isomorphic to the original input. Therefore, Kripke completeness is just \cfC\cfA\cfV-completeness.

In this way, we are led to the first of the two angles on Kripke incompleteness discussed in the introduction:  the semantic angle.  Given that the properties \cfC, \cfA, and \cfV\ are independent of each other, will arbitrary combinations of these three lead to distinct notions of completeness, each more general than Kripke completeness but less general than algebraic completeness?  Or is the propositional modal language too coarse to care about  differences between all or at least some of these semantics? And how about other notions contained somewhere in between? For example, \cfC\ can be weakened to $\omega$-completeness ($\omega\mathcal{C}$), i.e., closure under \emph{countable} meets and joins. One can then ask if there are logics that are $\omega\mathcal{C}$-complete but not \cfC-complete. For another important example, consider the property \cfT\ of \emph{admissibility of residuals/conjugates}.  Recall that $\goth{A}$
\emph{admits conjugates} if there is a function
$\resoper\colon \goth{A} \mapsto \goth{A}$ such that for every $a, b
\in \goth{A}$, $a \wedge \Diamond b = \bot$ iff $\resoper a
\wedge b = \bot$. Alternatively, we can say that an algebra
\emph{admits residuals} if there is a function
$\resdual\colon \goth{A} \mapsto \goth{A}$ such that for every $a, b
\in \goth{A}$, $a \leq \resdual b$ iff $\Diamond a
\leq b$. These two definitions are equivalent, taking $\resdual a =
\neg\resoper\neg a$.\footnote{Observe that we do not require that residuals
are term-definable. This is the difference between \baos\
admitting residuals and Jipsen's \citeyearpar{Jipsen93} \emph{residuated} \baos.} Some well-known facts include $\cfT \subseteq \cfV$ (i.e., admissibility of residuals implies complete additivity) and  \clfr{CT} = \clfr{CV} (i.e., in the presence of lattice-completeness, the converse implication also holds). Once again, one may ask: how does \cfT-completeness relate to all the other completeness notions?

As mentioned in \S~\ref{sec:intro}, a systematic investigation into these questions has been undertaken by Litak \citeyearpar{Litak2004,Litak2005,Litak2005b,Litak2008}, unifying, expanding, and building on earlier results by Thomason, Fine, Gerson, van Benthem, Chagrova, Chagrov, Wolter, Zakharyaschev, Venema, and  other researchers. There is no place here to discuss most of the results in detail, but an executive summary of those most relevant for the present paper is as follows:

\begin{itemize}
\item almost any conceivable combination of the above properties of \baos\ leads to a distinct notion of completeness. In other words, for almost any pair of such combinations,  there is a logic complete in one sense, but not in the other;
\item the Blok Dichotomy (\S~\ref{sec:Blok}), the  result of Wim Blok showing that Kripke incompleteness is in a certain mathematical sense the norm rather than an exception among normal modal logics, generalizes to most of these weaker notions of completeness;
\item many of these notions admit syntactic characterizations in terms of conservativity of certain types of extensions, e.g., \cfA\cfV-completeness in terms of conservativity of minimal nominal extensions (\S~\ref{subsec:hybrid}), \cfT-completeness in terms of conservativity of minimal tense extension (\S~\ref{subsec:tense}), or $\omega\mathcal{C}$-completeness in terms of conservativity of minimal infinitary extensions with countable conjunctions and disjunctions.
\end{itemize}

The possibility of \cfV-incompleteness, however, was left completely open, and in fact it was the sole reason for the ``\emph{almost} any'' hedge above (\citealt{Litak2004}, \citealt[Ch. 9]{Litak2005b}, \citealt[\S~7]{Litak2008}). Discussing this line of research in the \emph{Handbook of Modal Logic},  Venema \citeyearpar[\S~6.1]{Venema2007} singled out a slightly stronger version of the same question: whether there are $\mathcal{V}$-inconsistent logics, i.e., normal modal logics that are not sound over any $\mathcal{V}$-\bao.

Why was this question so puzzling? First of all, note that while a free algebra on infinitely many generators in any variety of \baos\ can never be lattice-complete or atomic, it can be completely additive.\footnote{\label{fn:FreeAlgebras}For a characterization of when the Lindenbaum-Tarski algebra of a normal modal logic is completely additive, see \citealt[\S~7.2]{Holliday2015}. Ghilardi \citeyearpar{Ghilardi1995} showed that free algebras in the variety of all \baos\ are $\mathcal{T}$-\baos\ (and hence $\mathcal{V}$-\baos), while Bezhanishvili and Kurz \citeyearpar{Bezhanishvili2007} extended this to all varieties of \baos\ axiomatized by rank-1 formulas. Holliday \citeyearpar{Holliday2014} added an analogous result for $\nml{KT}$, $\nml{K4}$, $\nml{KD4}$, and $\nml{S4}$ (and reproved it for $\nml{K}$, $\nml{KD}$, and all extensions of $\nml{KB}$).} One can imagine that if \cfV\ is not inconsistent with freeness, then there might be a general way of turning any \bao\  into a completely additive one without changing the set of valid equations. But there are other ways in which complete additivity seemed somewhat intangible. Unlike its closest relatives \cfA\cfV\ and \cfT, for which van Benthem's logic $\vblog$ can be shown to be incomplete, \cfV\ did not seem definable in a language with a usable model theory.

Let us make this more precise. There is an obvious  first-order correspondence language for \baos, whose connectives can be written 
as  $\forall, \exists, \llimp, \lland, \llor, \llnot$ (to avoid notational clashes with \bao-terms and modal formulas). For classes of algebras definable in this language, one can even blur the distinction between the class itself and its defining formula.
\cfA, \clfr{AV}, and \cfT\ are first-order properties, in fact even  $\forall\exists\forall$-properties. But how could one define \cfV\ without using infinitary formulas (of unrestricted cardinality!) or a powerful second-order language (with \emph{full} rather than Henkin semantics)?

\section{$\cfV$ as an Elementary Class}\label{sec:R&V}

Surprisingly, complete additivity is in fact a \textit{first-order} property. To prove this, it will be helpful to use some abbreviations for describing relations between elements of a \bao\ $\mathfrak{A}$. Let lower-case letters range over elements of $\mathfrak{A}$ and define:
\medskip
\begin{center}
\begin{tabular}{>{$}r<{$}@{ \hspace{1cm} stands for \hspace{1cm} }>{$}l<{$}}
a \blid{b}  & a \neq \bot \,\lland\, a\leq b \\
\exists a \blid{b}\; \alpha & \exists a\, (a \blid{b} \,\lland\, \alpha) \\
\forall a \blid{b}\; \alpha & \forall a\, (a \blid{b} \llimp \alpha).
\end{tabular}
\end{center}
\medskip

Consider a property of \baos\ formulated in our correspondence language as follows: 

\begin{center}
$\clfr{R}\colon \qquad \forall a,b \,\big((a \wedge \Diamond b \neq \bot) \llimp \exists c \blid{b} \,\forall d \blid{c} \,(a \wedge \Diamond d \neq \bot)\big)$. 
\end{center}
The origin of this condition is in the duality theory for classes of \baos\ and \textit{possibility frames} in \citealt{Holliday2015}. There the condition is viewed as follows. Given a \bao\ $\mathfrak{A}$, define a binary relation $R_\Diamond$ on the universe of $\mathfrak{A}$ by: $aR_{\Diamond}c$ iff for all $ d\preceq c$, we have $a\wedge \Diamond d\neq \bot$. Then $\mathfrak{A}$ satisfies $\mathcal{R}$ iff whenever $a\wedge \Diamond b\neq \bot$, there is a $c\preceq b$ such that $aR_{\Diamond}c$. Any such \bao\ can be turned into a possibility frame with the accessibility relation provided by $R_{\Diamond}$ and the validity relation coinciding with  that of    the original algebra.

\begin{remark} One may find a contrapositive formulation of $\mathcal{R}$  intuitive: 
\[
\forall a,b\, ((\forall c\preceq b\,\exists d\preceq c \;\Diamond d\leq a)\Rightarrow  \Diamond b\leq a)
\]
  (cf. Lemma~\ref{lem:Rref} for a $\Box$ reformulation).
\end{remark}

We will now prove that $\mathcal{R}$ is equivalent to complete additivity. 

\begin{theo} \label{th:rinv}
$\clfr{R}$ implies $\cfV$.
\end{theo}

\begin{proof}
We prove this by contraposition.  Assume $\Ag$ is not a $\mathcal{V}$-\bao. This means there is a $B\subseteq\Ag$ such that $\bigvee B$ exists in $\Ag$, but there is an $a \in \Ag$ such that:
\begin{itemize}
\item[(i)] for all $b\in B$, $\Diamond b \leq a$
\item[(ii)] $\Diamond\bigvee B \not\leq a$.
\end{itemize}
By (ii), $\neg a\wedge \Diamond\bigvee B \neq\bot$, so in order to refute \clfr{R}, it is enough to show that 
\begin{equation*}
\forall c \blid{\bigvee B} \,\exists d \blid{c} \,(\neg a \wedge \Diamond d = \bot).
\end{equation*}
Pick a $c \blid{\bigvee B}$, so $c \wedge \bigvee B\neq\bot$. By the join-infinite distributive law
\[
c \wedge \bigvee B = \bigvee\{c \wedge b \mid b \in B\}
\]
 holding in any Boolean algebra, it follows that there is a $b \in B$ such that $d := b \wedge c$ is not $\bot$, so $d\blid{c}$. But then
$$
\Diamond d = \Diamond (b \wedge c) \leq \Diamond b \leq a
$$
by (i), so $\neg a\wedge \Diamond d=\bot$.
\end{proof}

\begin{theo}\label{thm:VtoR}
$\cfV$ implies $\clfr{R}$.
\end{theo}

\begin{proof}
Again we reason by contraposition. Assume for some $a,b \in \Ag$ that
\begin{equation} \label{eq:contraR}
(a \wedge \Diamond b \neq \bot) \lland \forall c \blid{b} \,\exists d \blid{c} \,(a\wedge \Diamond d = \bot).
\end{equation}
Consider $B := \{ d \blid{b} \mid a \wedge \Diamond d = \bot\}$. To refute $\mathcal{V}$ it is enough to show that
\begin{equation*}
\bigvee B = b,
\end{equation*}
for then $a\wedge \Diamond \bigvee B=a\wedge \Diamond b\neq\bot$, yet for all $d\in B$, $a\wedge \Diamond d=\bot$, which implies \[\Diamond \bigvee B\not = \bigvee\{\Diamond d\mid d\in B\}.\] 
By definition of $B$, $b$ is an upper bound of $B$, so we need only show that it is the least. Suppose there is an upper bound $b'$ of $B$ such that $b\not\leq b'$. Hence $c:= b\wedge \neg b'$ is not $\bot$, so $c \blid{b}$. Then  \refeq{eq:contraR} implies there is a $d\blid{c}$ such that $d\in B$. Therefore $d\leq \neg b'$, and since $b'$ is an upper bound of $B$, $d\leq b'$. But then $d=\bot$, contradicting $d \blid{c}$.
\end{proof}

\begin{remark}
$\mathcal{T}$  implies $\mathcal{R}$ in a more direct way: where $\resoper$ is the conjugate of $\Diamond$, take $c:= \resoper a \wedge b$, so $a\wedge\Diamond b\neq\bot$ implies $\resoper a \wedge b\neq\bot$ and hence $c\blid{b}$; then if $d \blid{c}$ were such that $a\wedge\Diamond d=\bot$, we would have  $\resoper a \wedge d=\bot$ and hence $d\leq \resoper a \leq -d$, contradicting $d\blid{c}$. We shall see a derivation of a similar form at the end of \S~\ref{subsec:PureModal}.
\end{remark}

\begin{corollary} \label{cor:reqv}
$\cfV = \clfr{R}$.
\end{corollary}

Not only does this show that $\mathcal{V}$ is a first-order property but furthermore that it is of a rather convenient syntactic shape: $\forall\exists\forall$. Such conditions are particularly convenient for reformulation as \emph{non-standard} inference rules, which we will discuss in \S~\ref{sec:newsyntax}. 

\begin{remark} In response to our proof that $\mathcal{V}$ is equivalent to the first-order property $\mathcal{R}$, Johan van Benthem (p.~c.) devised a proof of the first-orderness of $\mathcal{V}$ in the style of correspondence theory \citep{Benthem1980}.  First note that in the equality for $\mathcal{V}$, \[\Diamond\bigvee X=\bigvee\{\Diamond x\mid x\in X\},\]
the $\geq$ direction is immediate from the monotonicity of $\Diamond$. Thus, using the equivalence of  $\Diamond\bigvee X\leq\bigvee\{\Diamond x\mid x\in X\}$ and $\forall z\big( \bigvee\{\Diamond x\mid x\in X\}\leq z\Rightarrow \Diamond\bigvee X\leq z\big)$, we can rewrite $\mathcal{V}$ as the following sentence in the second-order language of  \baos{}:
\begin{equation} \forall y\,\forall z\,\forall X\big(\big(y=\bigvee X\;\&\; \forall x\, (x\in X\Rightarrow\Diamond x\leq z)\big)\Rightarrow \Diamond y\leq z\big),\label{jvb1}\end{equation}
where `$y=\bigvee X$' abbreviates the first-order sentence expressing that $y$ is the least upper bound of $X$. Since $X$ does not occur in the consequent of the outer conditional in (\ref{jvb1}), we can rewrite (\ref{jvb1}) as 
\begin{equation}\forall y\,\forall z\,\big( \exists X \big(y=\bigvee X\;\&\; \forall x\,(x\in X\Rightarrow\Diamond x\leq z)\big)\Rightarrow \Diamond y\leq z\big).\label{jb2}\end{equation}
Now we observe that the antecedent of the outer conditional in (\ref{jb2}) is equivalent to a first-order sentence. For there exists an $X$ as in the antecedent iff $y$ is the least upper bound of the following first-order definable set:
\[X_{y,z}:=\{x\in\mathfrak{A}\mid x\leq y\;\&\;\Diamond x\leq z\}.\]
One direction of the `iff' is immediate. For the other, if a set $X$ is such that $y=\bigvee X$ and $ \forall x\, (x\in X\Rightarrow\Diamond x\leq z)$, then $X\subseteq X_{y,z}$ and hence $\bigvee X\leq \bigvee X_{y,z}$. Then since $y=\bigvee X$ and $\bigvee X_{y,z}\leq y$, we have $y=\bigvee X_{y,z}$, and by definition of $X_{y,z}$, we have $ \forall x\, (x\in X_{y,z}\Rightarrow\Diamond x\leq z)$. Thus, if there is any witness for the $\exists X$, then $X_{y,z}$ is a witness. Hence (\ref{jb2}) can be equivalently rewritten as the following first-order sentence:
\begin{equation}
\forall y\, \forall z\,( y=\bigvee X_{y,z}\Rightarrow \Diamond y\leq z),\label{jvb4}
\end{equation}
where `$y=\bigvee X_{y,z}$' abbreviates the first-order sentence expressing that $y$ is the least upper bound of $X_{y,z}$. This completes the proof that $\mathcal{V}$ is first-order. Further manipulations are required to show that (\ref{jvb4}) is equivalent to $\mathcal{R}$ in particular.\end{remark}

\section{$\cfV$-Incompleteness}\label{sec:V-inc}

Using the equivalence of $\mathcal{V}$ and $\mathcal{R}$, we can now prove a result from which all of our $\mathcal{V}$-incompleteness theorems will derive. 

\begin{theo}\label{thm:great} Let $\mathfrak{A}$ be a \bao{}, $a\in\mathfrak{A}$, and $\zedm$ and $\ondm$ operators on $\mathfrak{A}$ (not necessarily in the signature) such that 
\begin{enumerate}
\item\label{great1} $\ondm$ is completely additive, and 
\item\label{great2} for any $x\in\mathfrak{A}$, $a\leq \onbx (\zebx(\zebx x \to x)\to x)$.
\end{enumerate}
Then $a\leq \onbx \bot$.
\end{theo}

\begin{proof} Suppose that $a \wedge\ondm\top\neq \bot$. Since $\ondm$ is completely additive, by Theorem \ref{thm:VtoR} we can apply condition $\mathcal{R}$ to $\ondm$ with $b:=\top$ to obtain:
\begin{equation}\exists c \neq \bot \,\forall d \blid{c} \, (a \wedge \ondm d \neq \bot).\label{Rapp}
\end{equation}
Pick such a $c$ and observe that if $c \leq \zedm (c \wedge \zebx \neg c)$, then (using this very inequality to substitute $\zedm (c \wedge \zebx \neg c)$ for the first occurrence of $c$ on the right-hand side) we get
$$c \leq \zedm (c \wedge \zebx \neg c) \leq \zedm (\zedm (c \wedge \zebx \neg c) \wedge \zebx \neg c) = \bot,$$ a contradiction. Hence we have that
$$c \wedge \zebx (\zebx \neg c \to \neg c) \neq \bot.$$
Then by taking $d$ in (\ref{Rapp}) to be $c \wedge \zebx (\zebx \neg c \to \neg c)$, we have 

$$a \wedge \ondm (c \wedge \zebx (\zebx \neg c \to \neg c)) \neq \bot,$$
which contradicts condition \ref{great2} in the statement of the theorem.
\end{proof}

The fact that we do not insist on these operators to be primitive allows a lot of flexibility in instantiating this theorem, as we will witness below in Theorems \ref{th:Vinc1}, \ref{th:glbinc}, \ref{th:glsb}, \ref{th:ver}, \ref{th:triv}, and  \ref{th:vblok}. Note that working with this more general notion of operator would in fact allow us to replace $a$ with $\top$ in the above statement without loss of generality: the present statement would follow after replacing $\onbx x$ with $\onbx_a x \deq a \to \onbx x$ (as this transformation preserves complete additivity).

\subsection{The van Benthem Logic} \label{ssec:vb}

We are now ready to prove that the logic $\vblog$ of \S~\ref{sec:VB}, the smallest normal modal logic containing $ \Box\Diamond\top\rightarrow \Box (\Box (\Box p\rightarrow p)\rightarrow p)$, is $\mathcal{V}$-incomplete. To state this result in a more general form, let us borrow notation from \cite{Cresswell84}: let $\nml{ID}$ be the logic of the van Benthem frame $\vbfr$ from Definition \ref{def:VBframe} (see \S~\ref{subsec:Decide} for an explanation of this name).

\begin{theo}\label{th:Vinc1} Any logic between $\vblog$ and $\mathsf{ID}$ is $\mathcal{V}$-incomplete.
\end{theo}
\begin{proof} By Theorem \ref{thm:great}, taking $\ondm=\zedm=\Diamond$ (and hence $\onbx=\zebx=\Box$) and $a:=\Box\Diamond\top$,  any $\mathcal{V}$-\bao\, that validates $\vblog$ is such that $\Box\Diamond\top\leq \Box\bot$, so it validates $\Box\Diamond\top\to\Box\bot$. However, by Lemma \ref{lem:2}, $\Box\Diamond\top\to\Box\bot\not\in\mathsf{L}$ for any logic $\mathsf{L}$ contained in $\mathsf{ID}$. Since $\vblog$ is contained in $\nml{ID}$ by Lemma \ref{lem:2}, the theorem follows. 
\end{proof}

We have now come a long way from the initial Kripke-incompleteness results, i.e., $\mathcal{CAV}$-incompleteness results, of the 1970s. It turns out that each of the properties $\mathcal{C}$ \citep{Litak2004}, $\mathcal{A}$ \citep{Venema2003}, and finally $\mathcal{V}$ gives rise to incompleteness by itself.

Next, we will show that there are syntactically consistent bimodal logics that are $\mathcal{V}$-\textit{inconsistent} in the sense that they are not even sound with respect to any $\mathcal{V}$-\bao. This is not possible in the unimodal case, since by Makinson's Theorem \citep{Makinson1971}, every normal unimodal logic is sound with respect to a Kripke frame---either the single reflexive point or the single irreflexive point---and hence with respect to a $\mathcal{CAV}$-\bao{}. 

For the $\mathcal{V}$-inconsistency result, consider a bimodal language with modalities $\nmbx$ and $\aubx$ and let $\nml{vBe}$ be the smallest normal  logic in this language containing the $\vblog$-axiom for $\nmbx$ and the axiom $\audm(\nmbx\nmdm\top\wedge \nmdm\top)$. In addition, let $\vbefr$ be the general frame for the same language that extends $\vbfr$ such that the accessibility relation for $\audm$ in $\vbefr$ is the universal relation on the frame. Finally, let $\nml{IDe}$ be the logic of $\vbefr$. 

\begin{theo} \label{th:vbinc}\
\begin{enumerate}
\item\label{vbinc1} Any logic extending $\nml{vBe}$ (in particular $\nml{IDe}$) is $\mathcal{V}$-inconsistent, yet $\nml{vBe}$ is consistent.
\item\label{vbinc2} Any logic extended by $\nml{IDe}$ (in particular $\nml{vBe}$) is \cfA-consistent, and $\nml{IDe}$ is \cfA-complete. 
\end{enumerate}
\end{theo}

\noindent Note that we will prove more powerful results in Theorem \ref{th:conpV} and Corollary \ref{cor:id}.

\begin{proof} For the $\mathcal{V}$-inconsistency of extensions of $\nml{vBe}$, by the proof of Theorem \ref{th:Vinc1}, every $\mathcal{V}$-\bao{} that validates the $\vblog$-axiom for $\nmbx$ is such that $\nmbx\nmdm\top\wedge \nmdm\top=\bot$, so $\audm (\nmbx\nmdm\top\wedge \nmdm\top)=\bot$. For the consistency of $\nml{vBe}$, observe that $\vbefr$ validates $\nml{vBe}$.

For part \ref{vbinc2}, observe that the \bao\ underlying $\vbefr$ is atomic. 
\end{proof}

The logic $\vblog$ was introduced by van Benthem to prove a point about modal incompleteness, which we have pushed all the way to $\mathcal{V}$-incompleteness. Theorem \ref{th:Vinc1} raises the question: are there also ``naturally occurring'' examples of $\mathcal{V}$-incomplete logics?

As soon as we have at least two modal operators at our disposal, the answer turns out to be an emphatic ``yes''.

\subsection{The Provability Logic $\nml{GLB}$}\label{subsec:GLB}

To motivate the main logic of this section, we recall that formulas of propositional modal logic can be translated into sentences of Peano Arithmetic ($\nml{PA}$) as follows: map each atomic  $p_n$ to a sentence of arithmetic, send the modal $\Box$ to the arithmetized provability predicate $\mathrm{Bew}$ of $\nml{PA}$, and make the translation commute with the Boolean connectives in the obvious way. Solovay \citeyearpar{Solovay1976} showed that the modal logic $\nml{GL}$, the smallest normal modal logic containing the L\"{o}b axiom $\Box (\Box p\to p)\to \Box p$ is \textit{arithmetically sound and complete}: a modal formula $\varphi$ is a theorem of $\nml{GL}$ iff for every mapping of atomic sentences $p_n$ to sentences of arithmetic, the induced arithmetic translation of $\varphi$ is a theorem of $\nml{PA}$. Thus, $\nml{GL}$ captures the logic of the provability predicate of $\nml{PA}$. 

Japaridze \citeyearpar{Japaridze1988} introduced a polymodal extension of $\nml{GL}$, the bimodal version of which we will treat here. Let us interpret a bimodal language with operators $[0]$ and $[1]$ in the language of $\nml{PA}$ by sending $[0]$ to the provability predicate $\mathrm{Bew}$ of $\nml{PA}$ as before and sending $[1]$ to a predicate $\omega\mathrm{Bew}$ encoding \textit{provability from $\nml{PA}$ with one application of the $\omega$-rule}.\footnote{Equivalently, $\omega\mathrm{Bew}$ encodes provability from $\nml{PA}$ together with all $\Pi^0_1$ arithmetical truths.} A sentence $\sigma$ of arithmetic is provable in $\nml{PA}$ with one application of the $\omega$-rule if for some formula $\varphi(x)$, $\nml{PA}$ proves $\forall x\varphi(x)\to \sigma$ and proves $\varphi(\mathrm{n})$ for every numeral $\mathrm{n}$. The bimodal system that captures the combined logic of provability and $\omega$-provability in $\nml{PA}$ is the smallest normal bimodal logic containing the following axioms:
\begin{itemize}
\item[(i)] $\numbx{n}(\numbx{n}p\to p)\to \numbx{n}p$ for $n=0,1$;
\item[(ii)] $\zebx p\to \onbx p$;
\item[(iii)] $\zedm p\to \onbx\langle 0\rangle p$.
\end{itemize}
Japaridze \citeyearpar{Japaridze1988}  proved that this logic, now known as $\nml{GLB}$ \citep{Boolos1993}, is arithmetically sound and complete in the sense analogous to that of $\nml{GL}$ above.

While $\nml{GL}$ is Kripke complete \citep{Segerberg1971}, $\nml{GLB}$ is Kripke \textit{incomplete} (\citealt{Japaridze1988}, see also~\citealt{Boolos1993}, p.~194).   To see this, recall that over Kripke frames, the L\"{o}b axiom (i) for $\numbx{n}$ corresponds to the associated accessibility relation $R_n$ being transitive and Noetherian (conversely well-founded); axiom (ii), which we can equivalently take in the diamond form $\langle 1\rangle q\to \langle 0\rangle q$, corresponds to the property that $xR_1 y$ implies $xR_0 y$; and axiom (iii) corresponds to the property that if $xR_0y$ and $xR_1y'$, then $y'R_0 y$. 
 
Such a combination of axioms, however, makes $R_1$ an empty relation. Suppose there are $x,y$ for which $xR_1y$ in some Kripke frame for $\nml{GLB}$. Hence $xR_0y$ by the property corresponding to axiom (ii). Then using  $xR_0 y$, $xR_1 y$, and the property corresponding to axiom (iii), we obtain $yR_0 y$, which contradicts the Noetherianity of $R_0$ given by the L\"{o}b axiom for $\zebx$ (note that the argument does not use the L\"{o}b axiom for $\onbx$). But Japaridze's arithmetical soundness theorem for $\nml{GLB}$ shows that $\onbx\bot$ is not a theorem of $\nml{GLB}$ (a semantic argument can be extracted from Theorem \ref{th:glbe} below).  Thus, $\onbx\bot$ is a non-theorem of $\nml{GLB}$ that is valid in all Kripke frames for $\nml{GLB}$.

Viewed algebraically, Japaridze's Kripke-incompleteness result shows that $\nml{GLB}$ is $\mathcal{CAV}$-incomplete. Beklemishev et al. \citeyearpar{Beklemishev2010} show that $\nml{GLB}$ is complete with respect to a class of topological spaces, which implies that it is $\mathcal{CA}$-complete. In light of this result, it is a natural question whether the $\mathcal{CAV}$-incompleteness of $\nml{GLB}$ is due to the interaction of $\mathcal{V}$ with the other properties or whether it is due to the property $\mathcal{V}$ by itself. Using the equivalence of $\mathcal{V}$ and $\mathcal{R}$ from \S~\ref{sec:R&V}, we are able to answer this question and show that $\mathcal{V}$ by itself is to blame.

\begin{theo} \label{th:glbinc} The logic $\nml{GLB}$ is $\mathcal{V}$-incomplete.
\end{theo}

\begin{proof} First, for any \bao\ $\mathfrak{A}$ validating $\nml{GLB}$ and any $x\in\mathfrak{A}$, we have

\[\onbx (\zebx(\zebx x \to x)\to x)=\top,\] 
because 
\begin{eqnarray*}
\langle 1\rangle (\zebx (\zebx x \to x) \wedge \neg x) & \leq & \langle 1\rangle (\zebx x \wedge \neg x)\quad \mbox{by the L\"{o}b axiom for }\zebx\\
& \leq & \langle 1\rangle \zebx x \wedge \langle 1\rangle \neg x \\
& \leq & \langle 1\rangle \zebx x \wedge \langle 0\rangle \neg x\quad\mbox{by axiom (ii) of }\nml{GLB} \\
& \leq & \langle 1\rangle \zebx x \wedge \onbx\langle 0\rangle \neg x\quad\mbox{by axiom (iii) of }\nml{GLB}\\
& \leq & \bot.
\end{eqnarray*}
Thus, if $\mathfrak{A}$ is a $\mathcal{V}$-\bao, then we can apply Theorem \ref{thm:great} with $a:=\top$ to obtain $\top\leq \onbx\bot$. Yet as noted above, $\onbx\bot\not\in\nml{GLB}$.\end{proof}
Note that this proof uses the complete additivity only of the $[1]$ operator, so it shows that $\nml{GLB}$ is in fact $\mathcal{V}_1$-incomplete in the sense of Remark \ref{rem:V_n}. Similar remarks apply to later results involving variants of $\nml{GLB}$, though we will not mention this again.

We can get still more mileage out of this result by following the pattern of Theorem \ref{th:vbinc} to obtain a somewhat more ``natural'' logic  answering the \cfV-inconsistency question of Venema \citeyearpar[\S~6.1]{Venema2007}. 
 As in the case of $\vblog$, the problematic formula derivable over \cfV-\baos\ for $\nml{GLB}$ is variable free. Therefore, let us define $\nml{GLBe}$ as the smallest normal modal logic in the language with \emph{three} modalities $\zebx, \onbx$ and $\aubx$ containing the $\nml{GLB}$ axioms for $\zebx$ and $\onbx$ as above as well as the axiom $\audm\langle1\rangle\top$. 
 
\begin{theo} \label{th:glbe}
The logic $\nml{GLBe}$ is consistent (and even sound with respect to the ordinal version of the topological $\nml{GLB}$-semantics of \citealt{Beklemishev2010,Beklemishev2011}) but $\mathcal{V}$-inconsistent.
\end{theo}

\begin{proof}[Proof sketch]
Given  Theorem \ref{th:glbinc}, it is enough to prove the parenthetical claim. Recall that the topological semantics for $\nml{GLB}$ is defined in terms of spaces with two suitably related \emph{scattered} topologies; diamonds are interpreted by the  \emph{derived set operator} rather than Kuratowski's closure operator.\footnote{For all the notions undefined in this proof, see \citealt{BeklemishevG14}.} Ordinals provide particularly important instances of $\nml{GLB}$-spaces, with $\zebx$ interpreted by the \emph{order topology} and $\onbx$ by the  \emph{club topology}. In order to  extend this semantics to $\nml{GLBe}$, it is enough to interpret $\aubx$, e.g., with the interior operator of the trivial topology on any $\alpha \geq \aleph_\omega$. \end{proof}

One may hope to find not only a good topological interpretation but also a good arithmetical interpretation of $\audm$ validating $\audm\langle 1\rangle\top$. This would make $\nml{GLBe}$ a natural example of \cfV-inconsistency among normal logics. As things stand now, we have at least  an extraordinarily natural example of \cfV-incompleteness in Theorem \ref{th:glbinc}. And for \cfV-inconsistency, we can do better than $\nml{GLBe}$ in the \emph{quasi-normal} realm of \S~\ref{subsec:nonnormal}.

\subsection{\cfV-Inconsistency of $\nml{GLSB}$} \label{subsec:nonnormal}

If one is willing to broaden somewhat the setup of the present paper, a very natural example of \cfV-inconsistency can be found among modal logics without the necessitation rule. In analogy with \S~\ref{sec:VB}, given a normal modal logic $\nml{L} \supseteq \nml{K}$, we can define $\Sigma \vdash_\nml{L}^{mu} \varphi$ to mean that $\varphi$ belongs to the closure of $\nml{L}\cup \Sigma$ under modus ponens and uniform substitution. A set of formulas closed under $\vdash_\nml{L}^{mu}$ is called a \emph{quasi-normal logic over $\nml{L}$} or simply a {quasi-normal  $\nml{L}$-logic}. If $\nml{L} = \nml{K}$, we simply use the name  \emph{quasi-normal logic}. These notions transfer without any changes to the polymodal setting. 

Algebraic semantics for these quasi-normal logics has been well investigated. One can find a standard presentation in, e.g., \citealt[Ch.~7]{Chagrov1997}; an early exhaustive discussion is provided by Blok and K\"ohler \citeyearpar{BlokK1983}, who indicate that the basic notion of a \emph{filtered modal algebra} is a special case of the notion of a \emph{matrix} that dates back to the pre-war work of the Warsaw school.\footnote{\label{footnote:alghist}``A well-known result, going back to the twenties, states that, under some reasonable assumptions, any logic can be characterized as the set of formulas satisfied by a matrix $\langle S, F\rangle$, where $S$ is an algebra of the appropriate type, and $F$ a subset of the domain \dots'' \citep[p.~941]{BlokK1983}. In this opening quote, Blok and K\"ohler  were presumably referring to  \cite{LukasiewiczTarski}. The exact references  and more history can be found, e.g., in \citealt[\S~1.2]{FontJP03a:sl}.} As discussed by \cite{Jansana2006}, non-normal modal logics were an important inspiration for Blok's later work on abstract algebraic logic, and more generally, such logics have been a major  source of examples and applications in the area (see \citealt{BlokP89:ams,Andreka2001,Czelakowski2001,Font06:sl,FontJP03a:sl,FontJP09:sl}).

Instead of reproducing the whole apparatus here, let us just recall what is most relevant for our purposes. In this section, we are only interested in quasi-normal logics where each box operator obeys the L\"ob axiom. \baos\ in which each dual operator validates the L\"ob axiom are called \emph{diagonalizable}  \baos\ (also known as \emph{Magari algebras}). 

Let $\mathfrak{A}$ be a \bao. Recall that a non-empty subset $F \subseteq \mathfrak{A}$  is called a \emph{filter} if
\begin{itemize}
\item for any $a, b \in \mathfrak{A}$, $a, b \in F$ iff $a \wedge b \in F$. 
\end{itemize}
$F$ is \emph{proper} if $F \neq \mathfrak{A}$. A maximal proper filter is called an \emph{ultrafliter}.

Let $\nml{L}$ be a \emph{quasi-normal L\"ob logic}, i.e., a quasi-normal, polymodal logic such that for any $\numdm{n}$ in the signature, the unimodal restriction of $\nml{L}$ to $\numdm{n}$ is an extension of $\nml{GL}$. We say that a pair $\langle \mathfrak{A}, F \rangle$ is a \emph{matrix for $\nml{L}$} if
\begin{itemize}
\item $\mathfrak{A}$ is a diagonalizable \bao\ and $F$ is a filter on it, and
\item for any $\varphi \in \nml{L}$ (i.e., any theorem of $\nml{L}$) and any valuation $\theta$ on $\mathfrak{A}$, $\hat{\theta}(\varphi) \in F$.
\end{itemize}
It is a standard fact (see the references above) that every quasi-normal L\"ob logic $\nml{L}$ is sound and complete with respect to its class of matrices. 

We can apply our terminology for \baos\ to matrices as well: $\langle \mathfrak{A}, F \rangle$ is a \cfC-, \cfA-, \cfV-matrix if $\mathfrak{A}$ is a \cfC-, \cfA-, \cfV-\bao.  A matrix is \emph{degenerate} if $F$ is not proper, i.e., if $F = \mathfrak{A}$. 

Quasi-normal modal logics arise naturally in the context of provability logic. Say that a formula $\varphi$ of the language of $\nml{GLB}$ is \textit{always true} if for every mapping of atomic formulas to sentences of arithmetic, the induced arithmetic translation of $\varphi$ is true in the standard model of arithmetic. Then we may ask for a bimodal logic $\nml{L}$ such that $\varphi\in\nml{L}$ iff $\varphi$ is always true. Clearly $\nml{L}$ must be an extension of $\nml{GLB}$ that is closed under modus ponens and uniform substitution and contains the axioms $[n]p\to p$. But then $\nml{L}$ cannot be closed under necessitation, since that would give us $[0]([0]\bot \to \bot)$ as a theorem, which is not an always true formula by G\"{o}del's Second Incompleteness Theorem. The desired logic $\nml{L}$ is the logic  $\nml{GLSB}$ \citep{Boolos1993}, which can be defined as the quasi-normal logic axiomatized over  $\nml{GLB}$ by $\onbx p \to p$. Japaridze \citeyearpar{Japaridze1988} proved that $\nml{GLSB}$ is exactly the logic of the always true sentences about provability and $\omega$-provability.  As it turns out, this logic also provides a quasi-normal example of \cfV-inconsistency.

\begin{theo} \label{th:glsb}
There exists no non-degenerate \cfV-matrix for $\nml{GLSB}$.
\end{theo}

\begin{proof}
Assume $\langle \mathfrak{A}, F \rangle$ is a \cfV-matrix for $\nml{GLSB}$. Then $F$ has to contain all instances of $\onbx\varphi \to \varphi$, in particular $\onbx\bot \to \bot = \ondm\top$. But as shown in the proof of Theorem \ref{th:glbinc}, $\top \leq \onbx\bot$ holds in any \cfV-\bao\ validating $\nml{GLB}$, so the matrix is degenerate.
\end{proof}

\begin{remark} One may wonder if we could obtain such a result in the unimodal setting. While it would be possible to obtain examples of \cfV-inconsistent  quasi-normal unimodal logics using, e.g., $\nml{vB}$ again, the natural candidate in one's mind is probably $\nml{GLS}$ \citep{Boolos1993}, the logic of always true sentences about plain provability in $\nml{PA}$. It is well known that this logic does not allow any Kripke semantics with \emph{distinguished worlds}, and van Benthem \citeyearpar[\S~2.4]{Benthem1979} explicitly suggests that a related second-order derivation inspired the axiom of $\nml{vB}$.  Interestingly,  $\nml{GLS}$ allows $\cfV$-matrices and even \cfC\cfA\cfV-matrices. The real reason why the logic is problematic for the \emph{distinguished worlds} semantics is that matrices associated with this semantics can only use \emph{principal} filters. One can show that each such matrix for $\nml{GLS}$ (not even necessarily a \cfV-matrix) must be degenerate.  \end{remark}

\section{Decidability and Complexity} \label{sec:DecIssues}

Nontrivial completeness notions such as the ones studied in this paper raise questions about their relationship to decidability. Are there decidable logics, perhaps even of low complexity, that fail to be complete in the given sense? Is it decidable whether a logic is complete in this sense? Finally, is the associated notion of consequence recursively axiomatizable, i.e., does it allow some decidable notion of proof? 

Well-developed metatheory of modal logics leads us to the conclusion that for \cfV-completeness, the answers to these questions are, respectively, \emph{yes}, \emph{no}, and \emph{yes}. In this section, we discuss these answers in more detail.\footnote{Thanks to Johan van Benthem for inspiring this section with questions about complexity.}

\subsection{Decidable \cfV-Incomplete Logics}
\label{subsec:Decide}

The results of \S~\ref{sec:V-inc} show that \cfV-incomplete logics can easily be decidable. The existence of Kripke-incomplete yet decidable logics is not a new observation; an early result of this kind was due to \cite{Cresswell84}, who showed that the logic of the frame $\vbfr$ is Kripke-incomplete but decidable---hence the name `$\nml{ID}$' for this logic. It follows by Theorem \ref{th:Vinc1} that $\nml{ID}$ is an example of a decidable \cfV-incomplete logic. (It is an open question, suggested to us by Johan van Benthem, whether $\vblog$ itself is decidable.) Theorem \ref{th:glbinc} provides another example, since $\nml{GLB}$ is known to be decidable \cite[p.~206]{Boolos1993}. This example is again perhaps more spectacular, given the motivation for $\nml{GLB}$.

These examples raise follow-up questions. First, one may ask for an example of a decidable \cfV-inconsistent logic. Second, one may ask for bounds on the computational complexity of \cfV-incomplete/inconsistent logics.

On the question of complexity, Cresswell's \citeyearpar{Cresswell84} decidability proof for $\nml{ID}$ uses an embedding into Rabin's \citeyearpar{Rabin69} $\nml{S2S}$, which is much too powerful to provide a meaningful complexity bound. One can normally do much better. For a rather striking example, Litak and Wolter \citeyearpar{LitakW05} show that all tense logics of linear time flows that are either finitely axiomatizable or $\bigcap$-irreducible are coNP-complete. This class contains many $\omega\cfC$-inconsistent logics, i.e., logics that are not sound with respect to any non-degenerate \bao\ closed under countable joins. And yet, it turns out that their complexity is no worse than that of the classical propositional calculus. Similarly, it is possible to show  that $\nml{ID}$ is coNP-complete (see Corollary \ref{cor:id}). In fact, however, one can prove a stronger result that at the same time answers the first of the above questions:  even \cfV-\emph{inconsistent} logics do not have to be more complex than the Boolean calculus itself.

\begin{theo} \label{th:conpV}
$\nml{IDe}$, the logic of the general frame $\vbefr$ introduced in \S~\ref{ssec:vb}, is coNP-complete and \cfV-inconsistent. Furthermore,   in $\nml{ZF} + \nml{BPI}$ (i.e., assuming the Boolean Prime Ideal Theorem) it can be shown to be \cfC\cfA-complete.
\end{theo} 

\begin{proof}[Proof sketch] \cfV-inconsistency was already stated in Theorem \ref{th:vbinc}.\ref{vbinc1}. 
For the complexity claim, the crucial part of the proof is  analogous  to Lemma 4.3(a) in \citealt{LitakW05}. In more detail, for any $m \in \omega$ define the  $m$-\emph{collapse}  $\vbefr_m$ of $\vbefr$ as the substructure induced by   $\{\infty + 1,\infty\} \cup \{m, \dots, 0\}$. For a formula $\varphi$, let  $l(\varphi)$ be the cardinality of $nsub(\varphi)$: the single-negation-closed set of  subformulas of $\varphi$.  The crucial observation is that the satisfiability of $\varphi$ in $\vbefr$  is equivalent to its satisfiability in the $l(\varphi)$-collapse under a $\varphi$-\emph{good} valuation, i.e., one such that  for any $\chi \in nsub(\varphi)$, if $\chi$ holds at $\infty$, then it also holds at some $n \leq l(\varphi)$. This is shown as follows: 
\begin{itemize}
\item If $\varphi$ is satisfied in  $\vbefr$ by a valuation $V$, then for each $\chi \in nsub(\varphi)$ satisfied in $[\infty,0]$, one picks the maximal\footnote{Note this is the only difference with respect to \citealt{LitakW05}, where one needed to pick both  maximal and \emph{minimal} elements (and thus duplicate the size of the model) due to the presence of past modalities. Our simplified proof has the flavor of the selection-of-points argument for $\nml{GL.3}$.}  point in this interval where $\chi$ holds (note this is never $\infty$!). One then defines a morphism from $\vbefr_{l(\varphi)}$ to $\vbefr$ such that all these maximal points are in the codomain, which yields a valuation in $\vbefr_{l(\varphi)}$ via the inverse image. By definition, the obtained valuation is $\varphi$-good, and one can show inductively that satisfiability of formulas in $nsub(\varphi)$ is preserved and reflected.
\item Conversely, given a good valuation in  $\vbefr_{l(\varphi)}$, one takes this collapse as a substructure of $\vbefr$  (identifying corresponding ``infinity points'' and points with corresponding natural indices) and extends the valuation to the whole of $\vbefr$ by copying the values of all propositional variables from $\infty$ to all remaining natural numbers. Preservation and reflection of satisfiability  of formulas in $nsub(\varphi)$ is ensured by the $\varphi$-goodness of the valuation.
\end{itemize}
  The proof of the complexity claim is finished in the same way as for Theorem 2.1(ii) in \citealt{LitakW05}: given any $\varphi$, we simply guess a valuation in the $l(\varphi)$-collapse, whose size is bounded by $l(\varphi) + 3$, i.e., polynomial  in $\varphi$. Checking  its $\varphi$-goodness and the satisfaction of $\varphi$ itself under this valuation can be done in polynomial time. And coNP-hardness does not require much justification: $\nml{IDe}$ contains the propositional calculus.
  
  Finally, regarding \cfC\cfA-completeness, assuming $\nml{BPI}$: one can follow Wolter \citeyearpar[\S~4.6]{Wolter1993}  and Litak \citeyearpar[\S~4.4]{Litak2005} and modify $\vbefr = \langle W, R, \admis{W} \rangle$ to $\vbefr'  = \langle W, R, \admis{W}' \rangle$, where $\admis{W}'$ is obtained by fixing a non-principal ultrafilter $U$ over $\mathbb{N}$  and setting \[\admis{W}'=\{X\subseteq W\mid X \cap \mathbb{N} \not\in U \text{  and }\infty\not\in X\}\cup \{X\subseteq W\mid X \cap \mathbb{N} \in U \text{  and }\infty\in X\}.\] All that one needs to finish the proof  is to show that $\vbefr$ and $\vbefr'$ satisfy the same formulas, and this can be done by extending the equivalence proved by the two bullet points above to the following equivalence: $\varphi$ is satisfiable in $\vbefr$  iff  $\varphi$ is satisfiable in  $\vbefr_{l(\varphi)}$ under a $\varphi$-\emph{good} valuation iff $\varphi$ is satisfiable in $\vbefr'$.
  \end{proof} 

\begin{corollary} \label{cor:id}
$\nml{ID}$, the unimodal logic investigated in \citealt{Cresswell84},  is coNP-complete and \cfV-incomplete. Furthermore, in $\nml{ZF} + \nml{BPI}$ (i.e., assuming the Boolean Prime Ideal Theorem) it can be shown to be \cfC\cfA-complete.
\end{corollary}

\begin{proof}
This is a straightforward corollary of Theorems \ref{th:Vinc1}  and  \ref{th:conpV}. The $\mathcal{CA}$-completeness claim follows from the fact that the unimodal fragment of $\nml{IDe}$ is exactly $\nml{ID}$; it can also be shown directly as in the proof of Theorem \ref{th:conpV} (just neglecting the universal modality).
\end{proof}

This improvement on \citealt{Cresswell84} shows why using methods as powerful as embedding into Rabin's \citeyearpar{Rabin69} $\nml{S2S}$  to show decidability should be regarded as either the first step or the last resort. More tailored methods can  yield dramatically improved complexity bounds and as a bonus help to establish en route additional completeness results. Furthermore, ``tailored'' should not be taken to mean ``applicable to a single isolated system".  These points are illustrated, again, by tense logics of linear time flows. The methods used in the general complexity result of \cite{LitakW05} were developed in the course of an earlier thorough investigation of the lattice of these logics  by \cite{Wolter96,Wolter96:properties}. As pointed out in \citealt[Chapter 8]{Litak2005}, these methods also allow us to show \cfA\cfT-completeness of all such logics and to find an example of a logic that is $\mathcal{\kappa C}$-complete for any cardinal $\kappa$, yet \cfC-inconsistent.\footnote{Wolter and Zakharyaschev \citeyearpar[\S\S~6.2 \& 7]{WolterZ06} present the results of \cite{Wolter96,Wolter96:properties} and \cite{LitakW05} as examples that the \emph{Big Programme} or \emph{globalist's dream}, whose general failure we mention in \S~\ref{subsec:RecV} below, may fare better in restricted lattices of logics.} As in Theorem \ref{th:conpV} and Corollary \ref{cor:id}, fine-grained investigation of complexity and completeness seem to go hand in hand.

\begin{remark}We could take a cue from the tense example and generalize the results in this section to a broader class of \cfV-incomplete logics. For example, Wolter \citeyearpar[\S4.3]{Wolter1993}  defines a whole chain of similar logics. However, doing so would require a somewhat better motivation for classes of logics obtained in this way.\footnote{\cite{Wolter1993} was investigating the fine structure and properties of the lattice of subframe logics, a subject beyond our interest here.} We thus suggest just one modification: the point $\infty + 1$ in the frame $\vbefr$ and the antecedent  $\Box\Diamond\top$ in the $\vblog$-axiom are of little use in the bimodal case. As explicitly discussed by \cite{Benthem1979}, these contraptions were only used to get things working in the normal unimodal setting. Thus, if we give up the goal of staying as close as possible to  \cite{Benthem1979} and \cite{Cresswell84}, then we could modify $\vbefr$ and $\nml{IDe}$ accordingly: removing $\infty+1$, interpreting $\audm$ by $W \times \{\infty\}$ rather than the universal relation, and noting that $\aubx(\Box (\Box p\rightarrow p)\rightarrow p)$ and $\audm\top$ happen to be valid in the resulting general frame. The proof of \cfV-inconsistency would then just be a direct application of Theorem \ref{thm:great}. We leave adapting Theorem \ref{th:conpV} to this example as an exercise for the reader.\end{remark}

To close the discussion of complexity, let us note that $\nml{GLB}$, our other flagship example of \cfV-incompleteness,  is PSPACE-complete, as is the case with many other provability logics: on the one hand, the one-variable fragment of $\nml{GL}$ alone is PSPACE-complete \citep{Chagrov2003,Svejdar2003}, and on the other hand even the polymodal extension $\nml{GLP}$ of $\nml{GLB}$ with $\omega$-many modalities remains in PSPACE \citep{Shapirovsky2008}.

\subsection{Undecidability of \cfV-Completeness} \label{subsec:UndecV}

A more general question is whether the property of \cfV-completeness itself is decidable. Recall that \cite{Thomason1982} showed that it is undecidable whether a given (finite set of) axiom(s) axiomatizes a Kripke-complete logic. Further results along these lines can be found in, e.g., \citealt{KrachtW99} or \citealt{Chagrov1997}. Refinements include the discussion of decidability of Kripke completeness in smaller lattices of logics. In particular, \cite{Chagrov1990} (see also \citealt[Theorem 17.19]{Chagrov1997}) shows that Kripke completeness is undecidable over $\nml{GL}$. 

Adopting Thomason's technique to show the undecidability of \cfV-completeness for polymodal logics in general is straightforward; as discussed by  \cite{KrachtW99}, all that one really needs to show is its reflection under fusions.\footnote{We refer the reader to, e.g., \citealt{KrachtW91} for more on preservation and reflection of properties of modal logics by fusions.}  

\newcommand{\fus}[2]{\nml{#1} \otimes \nml{#2}}

In fact, as highlighted by Kracht and Wolter \citeyearpar[p.~137]{KrachtW99}, Thomason's methodology is of sweeping generality, and one could even call it a modal variant of Rice's Theorem.\footnote{Recall that Rice's Theorem \citeyearpar{Rice53} states that \emph{every nontrivial property of recursively axiomatizable languages is undecidable} \cite[Theorem 9.11]{HopcroftMU03}. Note that \cite{ChagrovZ93:jsl} take a somewhat different perspective on the relationship of Rice's Theorem to modal logic, focusing mostly on unimodal systems.} Let us isolate this result. Recall that the \emph{fusion} $\fus{L_1}{L_2}$ of two normal logics $\nml{L_1}$ and $\nml{L_2}$ formulated in disjoint modal signatures is obtained by taking the sum of their theorems and closing it under the axioms and rules of polymodal $\nml{K}$ in the combined signature. This notion is extended to arbitrary  logics with possibly overlapping signatures by trivial renaming of operators to ensure disjointness. We say that a property of polymodal normal logics $P$ is 
\begin{itemize}
\item \emph{reflected by fusions} if  $\nml{L_1}\in P$ and $\nml{L_2} \in P$ whenever $\fus{L_1}{L_2} \in P$, and
\item \emph{finitely nontrivial} if the inconsistent logic in every signature has $P$ but there is at least one finitely axiomatizable $\nml{L} \not\in P$.
\end{itemize}

\begin{theo}[\cite{Thomason1982,KrachtW99}] \label{th:modalrice}
Any finitely nontrivial property of polymodal logics that is reflected by fusions is undecidable for normal modal logics with at least three operators. 
\end{theo}  

\begin{proof}
The proof boils down to Thomason's \citeyearpar{Thomason1982} result that it is undecidable whether a formula with at least two modal operators axiomatizes a consistent normal logic. 
\end{proof}
 
For the next statement, recall  the notion of a (modal) \emph{reduct} of a \bao. Assume  that $\Ag$ is a \bao\ for the signature where the set of operators extending Boolean operations is $S_I = \{\Diamond_i\}_{i \in I}$. Consider any subset $J \subseteq I$. Clearly, $\Ag$ can be modified to a \bao\ $\Ag^J$ whose set of operators extending Boolean operations is $S_J = \{\Diamond_i\}_{i \in J}$, simply by forgetting the interpretations of operators with indices in $I - J$. Such an $\Ag^J$ is a (\emph{modal}) \emph{reduct}  of $\Ag$. Consider a property $\clofr{X}$ of \baos. If for any $\Ag \in \clofr{X}$ and any $J \subseteq I$, we have  $\Ag^J \in \clofr{X}$, we say $\clofr{X}$ is \emph{preserved by modal reducts}. It is not a high bar to clear; in fact, it would be hard to come up with a natural property \emph{not} preserved by reducts.
 
\begin{corollary}
Let $\clofr{X}$ be any property of \baos\ that is preserved by modal reducts. If the property of $\clofr{X}$-completeness is finitely nontrivial, then it is undecidable for  normal modal logics with at least three operators.
\end{corollary}  

\begin{proof}
Preservation of $\clofr{X}$ by modal reducts implies that the associated notion of completeness is reflected by fusions.
\end{proof}

\begin{corollary}
\cfV-completeness is undecidable for  normal modal logics with at least three operators (recall Remark \ref{rem:V_n}). 
\end{corollary} 

\begin{proof}
The property $\mathcal{V}$ of complete additivity is preserved by modal reducts, and $\mathcal{V}$-completeness is finitely nontrivial. 
\end{proof}

\begin{corollary}
\cfV-completeness is undecidable for any class of normal unimodal logics containing the Thomason-simulation logic $\nml{Sim}(3)$ \cite[\S~6.8]{Kracht99} (also denoted by $\nml{Sim}(3,1)$ in \citealt[\S~9]{KrachtW99}).
\end{corollary}

\begin{proof}[Proof sketch]
 As discussed by \cite{Thomason1982} and \cite{KrachtW99}, one needs to establish that the property in question \emph{transfers under simulations} such as those presented by \cite{Thomason74:reduction1,Thomason74:reduction2} and \cite{KrachtW99}. This is indeed the case for the property  $\mathcal{V}$.  The details of the argument are analogous to those of the proof of Theorem 2.51 in \citealt{Holliday2015}, which is dealing with \cfC\cfV\ rather than \cfV; lattice-completeness plays no role in the argument. 
\end{proof}

The question of the decidability of \cfV-completeness in more restricted lattices of logics is much less trivial. We have no evidence of the existence of \cfV-incomplete logics in the lattice of extensions of unimodal $\nml{K4}$, much less $\nml{GL}$. Thus, it might be that this property is decidable for such logics in a degenerate sense, which would provide a rather dramatic contrast with \citealt{Chagrov1990}.
 
\subsection{Recursive Axiomatizability of \cfV-Consequence} \label{subsec:RecV}

Results such as Theorem \ref{th:modalrice} show the ultimate unfeasibility of what Wolter and Zakharyaschev \citeyearpar{WolterZ06} call  the \emph{Big Programme} or \emph{globalist's dream}\footnote{\label{ft:globalist} As Wolter and Zakharyaschev \citeyearpar[p.~428]{WolterZ06} describe it, ``Although not formulated explicitly, the `globalist's' dream research programme was to
develop a mathematical machinery that could \dots\ [e.g.,] given a modal logic in the form of a finite set of axioms and inference rules, characterise
the (simplest, smallest, largest, etc.) class of models/structures with respect
to which this logic is sound and complete,  decide in an effective way whether it has
important properties\dots [such as decidability,
compactness, interpolation] and determine its computational complexity.''} and  claim to be an implicit motivation for much of the early work in modal logic.  An earlier blow to this program was delivered by Thomason's \citeyearpar{Thomason1975c} reduction of monadic second-order consequence to modal Kripke-frame consequence.  To borrow a phrase from \cite{Blackburn2001}, this result showed that no \emph{strengthening  of our deductive apparatus can eliminate frame incompleteness}. It is not just that \cfC\cfA\cfV-consequence fails to be recursively axiomatizable and hence does not allow any decidable notion of proof; it is beyond the entire arithmetical hierarchy.

We have seen (and will see even better in \S~\ref{sec:Blok}) that many negative results regarding Kripke completeness hold as well for \cfV-completeness, just as they turned out to hold for other algebraically-inspired weak notions of completeness \citep{Litak2005,Litak2005b,Litak2008}. However, there is also an important contrast: while the class $\mathcal{CAV}$ is not a first-order definable class of \baos, we have seen that $\mathcal{V}$, like \cfA, \cfA\cfV, and \cfT, \textit{is} first-order definable, which gives us the following positive result in contrast to $\mathcal{CAV}$-consequence.

\begin{corollary} The consequence relation $\vDash_\mathcal{V}$ over $\mathcal{V}$-\baos\ is recursively axiomatizable.
\end{corollary}

The difference, of course, boils down to the strong completeness of first-order logic. The desired recursive axiomatization of the consequence relation and a decidable notion of proof is provided by using the obvious translation of the modal language into the algebraic correspondence language provided in the course of the algebraization process.\footnote{Obviously, this argument goes well beyond logics with the classical propositional base \citep{Rasiowa74:aatnl,BlokP89:ams,Andreka2001,Font06:sl,FontJP03a:sl,FontJP09:sl}.}

This, to be sure, is a suboptimal argument. One would like to avoid stepping so far outside the ordinary modal syntax, especially since deductions in the first-order metalanguage for \baos\ can involve sentences with arbitrary nesting of quantifiers. Indeed, one can do better for many FO-definable properties of \baos, especially those definable by $\forall\exists\forall$-sentences like \cfA\cfV\ (see \S~\ref{subsec:hybrid}), \cfT\ (see \S~\ref{subsec:tense}) or \cfA\ (see \citealt{Litak2006}). The general theme of finding better ``internalizations'' of consequence relations is the \emph{leitmotif}  of \S\S~\ref{sec:syntax}--\ref{sec:newsyntax}, with \S~\ref{sec:newsyntax} entirely devoted to \cfV-consequence. For more material on the ``internalization'' process, we also refer the reader to \citealt{H&L2016}.

\section{The Blok Dichotomy}\label{sec:Blok}

We recalled in \S~\ref{subsec:RecV} that Thomason's \citeyearpar{Thomason1975c} reduction of monadic second-order logic to modal Kripke-frame consequence was one of the two biggest blows  the 1970s delivered to what  \cite{WolterZ06} call the \emph{Big Programme} or \emph{globalist's dream} of the 1960s. The other one was delivered by Blok \citeyearpar{Blok1978}, with what has come to be called the  \emph{Blok Dichotomy}. An analysis of sub-Kripkean completeness notions shows that these two results have nothing whatsoever to do which each other: the Blok Dichotomy can still obtain even where no Thomason-style result holds. In this section, we will show that \cfV-completeness and \cfV-consequence is  perhaps the most dramatic case in point.\footnote{Previously known examples included \cfA\cfV\ (\citealt{Litak2005,Litak2008}) and \cfT\ (\citealt{Zak2001,Litak2005,Litak2008}), i.e., subclasses of \cfV, and $\clofr{\omega C}$ (\citealt{Litak2005,Litak2008}), which is not an FO-definable property.}

\subsection{Blok Dichotomy for Kripke Incompleteness} 

Our goal is to prove the Blok Dichotomy for $\mathcal{V}$-incompleteness. In order to explain what this means, we need to review some basic notions. 

Fine \citeyearpar{Fine1974b} defined the \textit{degree of Kripke incompleteness} of a normal modal logic $\nml{L}$ to be the cardinality of the set of normal modal logics $\nml{L}'$ such that $\nml{L}$ and $\nml{L}'$ are valid over exactly the same class of Kripke frames. Only one of the logics $\nml{L}'$ in that set can be Kripke complete, namely $\mathrm{Log}(\mathrm{Fr}(\nml{L}))$. If $\nml{L}$ is Kripke \textit{incomplete}, then its degree of Kripke incompleteness is $\geq 2$, since $\nml{L}$ and $\mathrm{Log}(\mathrm{Fr}(\nml{L}))$ are distinct logics valid over exactly the same Kripke frames. In this case  $\mathrm{Log}(\mathrm{Fr}(\nml{L}))$ is a Kripke complete logic with degree of incompleteness $\geq 2$. By contrast, the logic $\nml{K}$, which is the logic of the class of all frames, has degree of incompleteness 1. Such logics are said to be \textit{strictly Kripke complete}.

Earlier Kripke incompleteness results showed that there are normal modal logics with degree of Kripke incompleteness $\geq 2$, and Fine noted that a normal modal logic with degree of incompleteness $2^{\aleph_0}$ could be produced. Fine then asked what cardinalities may be degrees of Kripke incompleteness and whether there are any strictly complete normal modal logics other than $\nml{K}$.

Blok \citeyearpar{Blok1978} provided an exact and surprising answer to these questions: the only strictly complete normal modal logics can be characterized in terms of their occupying a special position in the lattice of all normal modal logics; and all other normal modal logics have degree of incompleteness $2^{\aleph_0}$, i.e., each one of them shares the same class of Kripke frames with $2^{\aleph_0}$ Kripke-incomplete logics. 

To say what the special position of the strictly complete logics is, we recall that the set of all normal modal logics, ordered by inclusion, forms a complete distributive lattice $\mathbb{L}$: the \textit{meet} of a family of logics is their intersection; the \textit{join} of a family of logics is the smallest normal modal logic that includes each logic in the family (which might strictly extend the union of the family); the bottom element of $\mathbb{L}$ is $\nml{K}$; and the top element is the set $\mathcal{L}$ of all formulas. In this lattice $\mathbb{L}$, a pair $\langle\nml{L}_1,\nml{L}_2\rangle$ of normal modal logics is called a \textit{splitting pair} iff for every logic $\nml{L}$ in $\mathbb{L}$, either $\nml{L}\subseteq\nml{L}_1$ or $\nml{L}_2\subseteq\nml{L}$, but not both, i.e.,  $\nml{L}_2\not\subseteq\nml{L}_1$. Thus, $\langle\nml{L}_1,\nml{L}_2\rangle$ splits $\mathbb{L}$ into two disjoint parts. Note that if $\langle\nml{L}_1,\nml{L}_2\rangle$ is a splitting pair, then either member of the pair determines the other uniquely. A logic $\nml{L}_2$ is a \textit{splitting} of $\mathbb{L}$ iff there is a logic $\nml{L}_1$ in $\mathbb{L}$ such that $\langle\nml{L}_1,\nml{L}_2\rangle$ is a splitting pair. For $\nml{L}_2$ to be a splitting is equivalent to it being \textit{completely join-irreducible} in the usual sense of lattice theory: $\nml{L}_2$ cannot be obtained as the join of a family of logics all of which are distinct from $\nml{L}_2$. Finally, a logic is a \textit{join-splitting} of $\mathbb{L}$ iff it is a join of a family of splittings, or equivalently, of completely join-irreducible logics. 

Two important facts about splitting logics were established by Blok en route to his main result. We will also need them below and hence we single out these two facts  as a separate lemma. Recall that a normal modal logic $\nml{L}$ is \textit{finitely approximable} iff there is a class $\mathsf{F}$ of finite Kripke frames such that $\nml{L}=\mathrm{Log}(\mathsf{F})$, which is equivalent to $\nml{L}$ having the \textit{finite model property} (see \citealt[Thm.~8.47]{Chagrov1997}).

\begin{lemma}[\citealt{Blok1978}]\  \label{lem:jsplfin}
\begin{enumerate}
\item\label{jsplfin1}  All join-splittings are finitely approximable;
\item\label{jsplfin2} If a finite Kripke frame $\ffr$ is cycle free, then  $\mathrm{Log}(\ffr)$ splits the lattice of normal modal logics:  there is an $\nml{L}_2$ such that $\langle\mathrm{Log}(\ffr),\nml{L}_2\rangle$ is a splitting pair.
\end{enumerate}
\end{lemma}

Blok's \citeyearpar{Blok1978} main result can now be stated precisely as follows.

\begin{theo}[Blok Dichotomy for Kripke frames]\label{thm:Blok} If a consistent normal modal logic is a join-splitting, then it is strictly Kripke complete; otherwise it has degree of Kripke incompleteness $2^{\aleph_0}$.
\end{theo}

\begin{proof}
Apart from the original proof by \cite{Blok1978}, one can  consult numerous more recent references, in particular \citealt{Chagrov1997,Chagrova1998,Kracht99,Zak2001,Litak2005b,RautenbergWZ06,WolterZ06,Litak2008}.  
\end{proof}

Given the duality between Kripke frames and $\mathcal{CAV}$-\baos{}, Theorem \ref{thm:Blok} can be equivalently stated in terms of degrees of $\mathcal{CAV}$-incompleteness. Below we will extend Theorem \ref{thm:Blok} to degrees of $\mathcal{V}$-incompleteness, thereby showing that neither of the properties $\mathcal{C}$ or $\mathcal{A}$ of the duals of Kripke frames are necessary for the Blok Dichotomy. 

\subsection{Maximal Consistent Logics}

The \textit{degree of $\mathcal{V}$-incompleteness} of a normal modal logic $\nml{L}$ is the cardinality of the set of normal modal logics $\nml{L}'$ such that $\nml{L}$ and $\nml{L}'$ are valid over exactly the same $\mathcal{V}$-\baos{}. 

Following the pattern of the proofs quoted in the proof of Theorem \ref{thm:Blok}, before we generalize the theorem itself, we consider two special cases and then use proofs of these special cases in the proof of the main theorem. Recall that by Makinson's Theorem \citep{Makinson1971}, the two maximal consistent unimodal normal logics are $\nml{Triv}$ (the logic of the reflexive point) and $\nml{Ver}$ (the logic of the irreflexive point).

\begin{theo} \label{th:ver}
 The degree of $\mathcal{V}$-incompleteness of $\nml{Ver}$ is continuum.
\end{theo}

All the details of the proof can be lifted directly from \citealt[\S~8]{Litak2008}, modulo corrected typos in that proof, adjusted notation, and \cfV\ replacing \cfT, \cfA\cfV, or the sum theoreof. For the sake of completeness, we reproduce the argument here.

\setlength{\unitlength}{1.2mm}

\begin{figure}
\begin{center}
\begin{picture}(90,40)

\put(5,15){\dashbox{1}(70,25)[tl]{}}
\put(9,35){transitive}

\multiput(21,20)(10,0){6}{\vector(1,0){8}}
\multiput(21,21)(10,0){2}{\vector(1,1){8}}
\multiput(31,29)(10,0){2}{\vector(1,-1){8}}
\put(51,21){\vector(1,1){8}}
\put(61,29){\vector(1,-1){8}}
\put(70,19){\line(0,-1){9}}
\put(70,10){\vector(-1,0){59}}
\put(10,11){\vector(0,1){8}}

\multiput(15,20)(1.25,0){3}{\circle*{0.3}}

\multiput(9,19)(10,0){8}{$\bullet$}
\multiput(29,29)(10,0){2}{$\bullet$}
\put(59,29){$\bullet$}
\put(9,9){$\bullet$}

\put(9,21.5){$b$} 
\put(19,17){$a_6$} \put(29,17){$a_5$} \put(39,17){$a_4$} \put(49,17){$a_3$} \put(59,17){$a_2$}

\put(69,22){$a_1$} \put(79,17){$a_0$}

\put(29,32){$a'_5$} \put(39,32){$a'_4$} \put(59,32){$a'_2$}

\put(9,7){$c$}

\end{picture}
\vspace{0.5cm}
\caption{The frame $\vbenthem_I$ for $I = \{2, 4, 5, \ldots\}$. \label{fig:modified-vb}}
\end{center}
\end{figure}

\begin{proof}
 Fix an arbitrary $I
\subseteq \omega - \{0, 1\}$. For $i > 1$, define 
\[
\doubled{a_i} \equals \left\{ \begin{array}{l@{\quad \mathrm{if}\quad}l}
\{a_i\} & i \not\in I, \\
\{a_i, a'_i\} & i \in I.
\end{array} \right.
\]
Now let $\vbenthem_I \equals  \langle W_I, R_I, \admis{F}_I \rangle$, where (see~Figure~\ref{fig:modified-vb}):
\begin{itemize}
\item $W_I \equals
\{b\}\cup\{c\} \cup \bigcup\limits_{k \in \omega}\doubled{a_k}$; 
\item $R_I := \{\langle c,b\rangle
\} \cup (\{b\} \times \bigcup\limits_{k \in \omega}\doubled{a_k})
\cup \bigcup\limits_{k > l > 1}(\doubled{a_k} \times \doubled{a_l}) \cup\{\langle a_1,a_0\rangle\}\cup\{\langle a_1,c\rangle\}$;
\item $\admis{F}_I$ consists of finite sets that do not contain $b$ and their
complements. 
\end{itemize}
Let $\vblogic_I$ denote the logic of
$\vbenthem_I$, i.e., the set of all modal formulas valid over
$\vbenthem_I$. We will prove that 
\begin{itemize}
\item distinct $I \subseteq \omega$ produce distinct $\vblogic_I$, but
\item all of these logics share the same class of \cfV-\baos\ with $\nml{Ver}$.
\end{itemize}
First, let us define the following sequences of formulas,\footnote{We use the convention that if $\xi(p)$ is a formula containing $p$ and $\mu$ is a formula, then $\xi(\mu)$ is the result of substituting $\mu$ for all occurrences of $p$ in $\xi(p)$.} which will provide names of points of $\vbenthem_I$:
\begin{eqnarray*}
\alpha_0(p) & \equals & p \\
\alpha_1(p) & \equals & \oper{} \alpha_0(p) \wedge \dual{}^2 \neg \alpha_0(p) \\
\alpha_{k+2}(p) & \equals & \oper{} \alpha_{k+1}(p) \wedge \dual{}^2 \neg \alpha_{k+1}(p) \wedge \oper{}\alpha_1(p) \\
\underline{a_i} & \equals & \alpha_i(\Box\bot) \\
\gamma(p) & \equals &
\Diamond^2\alpha_1(p)\wedge\neg\Diamond\alpha_1(p) \\ 
\underline{c} & \equals & \gamma(\Box\bot). 
\end{eqnarray*}
Moreover,  for arbitrary $\varphi$, let $\dualrel\varphi \equals \Box(\univdmd\underline{a_1} \to \varphi)$. Note that this is a normal modality, with its dual $\operrel\varphi := \Diamond(\univdmd\underline{a_1} \wedge \varphi)$. Recall that $\univdmd\varphi$ is $\varphi \vee \Diamond\varphi$ and its dual operator is $\univbox\varphi := \varphi \wedge \Box\varphi$. 

For arbitrary $I \subseteq \omega - \{0,1\}$, $k \in \omega$, and an arbitrary valuation $V$
on $\vbenthem_I$, we have that $V(\underline{a_k}) = \doubled{a_k}$ and
$V(\underline{c}) = \{c\}$. Hence, for arbitrary $i \in
\omega$, 
\begin{center}
$\vbenthem_I \vDash \Box(\underline{a_i} \imply p) \vee
\Box (\underline{a_i} \imply \neg p)$ 
 \quad iff \quad $i \not\in
I$, 
\end{center}
and thus $\vblogic_I = \vblogic_J$ iff  $I = J$.

Second, define
\begin{eqnarray*}
 \epsilon & \equals& \underline{a_0} \vee \univdmd^2 \underline{a_1} \\
  \zeta & \equals & \underline{a_1}  \to  \Diamond \underline{c} \\
  \eta & \equals & \underline{c}  \to 
  \Box(\dualrel(\dualrel p \to p) \to p).
\end{eqnarray*}
Let $\lambda$ be the conjunction of these three formulas. Now notice that:

\begin{itemize}
\item For arbitrary $I \subseteq \omega - \{0,1\}$, $\vbenthem_I \vDash \lambda$.
\item For arbitrary $\goth{A} \in \cfV$ with $\goth{A} \satisf \lambda$, $\goth{A}
\satisf \underline{a_0}$.
\end{itemize}
The first of these claims follows from the definitions of $\vbenthem_I$ and  $\underline{a_0}$, $\underline{a_1}$ and $\underline{c}$, so the second is all we need to finish the proof.  Assume for reductio that in a $\mathcal{V}$-\bao\ where $\lambda$ holds, we have that $\underline{a_0} \neq \top$. Observe that $\underline{a_1} \neq \bot$ by $\epsilon$ and hence
$\underline{c}\neq \bot$ by $\zeta$; and by definition, $\underline{c} \leq
\Diamond^2\underline{a_1}$.  But by picking $a:=\underline{c}$, $\zebx := \dualrel$, and $\onbx := \Box$ in Theorem \ref{thm:great}, we also get that
$\underline{c} \leq \Box\bot$. Thus, 
\[ \underline{c} \leq \Box\bot\wedge \Diamond^2\underline{a_1} \leq \Box\bot \wedge \Diamond \top = \bot, \]
a contradiction.\end{proof}

\begin{theo} \label{th:triv} 
 The degree of $\mathcal{V}$-incompleteness of $\nml{Triv}$ is continuum.
\end{theo}

\newcommand{\bole}[1]{\univbox^{#1}}%{\Box^{\leq#1}}
\newcommand{\dile}[1]{\univdmd^{#1}}%{\Diamond^{\leq#1}}

\begin{proof}
This time, we cannot use variable-free formulas. However, the rest of the proof will be remarkably similar to that of Theorem \ref{th:ver}, including the use of a family of frames $\{\vbenthem^\circ_I\}_{I \subseteq \omega}$, the only modification compared to $\{\vbenthem_I\}_{I \subseteq \omega}$ being that $a_0$ is now taken to be reflexive. In fact, the proof is entirely analogous to that of \citealt[Example 10.58]{Chagrov1997}, with the only minor differences stemming from a slightly more economical way we chose to define our frame (as in \citealt{Litak2008}). 

Note that if a formula 
\[
\zeta(p) \deq \dile{3} p \wedge  \dile{3}\neg p
\]
is  satisfied at some point of $\vbenthem^\circ_I$ under some valuation $V$, then 
\[
\alpha_1(p) \deq  (\Diamond\bole{3}p \wedge \dile{3}\neg p) \vee (\Diamond\bole{3}\neg p \wedge \dile{3} p)
\]
is satisfied at $a_1$ and nowhere else under the same valuation $V$. Formulas $\{\alpha_i(p)\}_{i \geq 1}$ and $\gamma(p)$ can now be taken verbatim from the proof of Theorem \ref{th:ver} and characterize corresponding points under $V$. 
Hence, for arbitrary $i \in
\omega$, 
\begin{center}
$\vbenthem^\circ_I \vDash \zeta(p) \to (\univbox^3(\alpha_i(p) \imply q) \vee
\univbox^3(\alpha_i(p) \imply \neg q))$ \quad iff \quad $i \not\in
I$, 
\end{center}
and thus $\vblogic^\circ_I = \vblogic^\circ_J$ iff $I = J$.

Now take $\lambda^\circ$ to be the conjunction of
\begin{eqnarray*}
 \epsilon^\circ & \equals & \zeta(p) \to \univdmd^3 \gamma(p) \\
  \eta^\circ & \equals & \gamma(p) \to 
  \Box(\dualrel(\dualrel q \to q) \to q),
\end{eqnarray*}
note its validity over arbitrary $\vblogic^\circ_I$, and use Theorem \ref{thm:great} to finish the proof in exactly the same way as we did in Theorem \ref{th:ver}.\end{proof}

\subsection{Blok Dichotomy for \cfV-Incompleteness}

We can now finally state and prove the general result for degrees of \cfV-incompleteness. 

\begin{theo} \label{th:vblok}
If a consistent normal modal logic $\nml{L}$ is not a join-splitting of the lattice of normal modal logics, then $\nml{L}$ has degree of $\mathcal{V}$-incompleteness $2^{\aleph_0}$.
\end{theo}

One could argue that there is no need to provide all the details of the proof. While Blok's original construction was unsuitable for  generalizations to most classes of algebras containing \cfC\cfA\cfV\ (see \citealt{Dziobiak78} and in particular \citealt{Litak2005b,Litak2008} for a detailed discussion), a strategy proposed by Chagrov and Zakharyaschev in the 1990's and followed by most of the ``recent references'' quoted in lieu of the proof of Theorem \ref{thm:Blok} is much more flexible.  In particular, \cite{Chagrov1997,Zak2001,Litak2005b,RautenbergWZ06,WolterZ06} and \cite{Litak2008} show how to use $\vblog$ and $\vbfr$ in proofs of the Blok Dichotomy and its generalizations to notions like \cfA\cfV- and \cfT-completeness. Apart from using the more general argument of \S~\ref{sec:V-inc}, there is no significant conceptual difference in the present version. However, we still include a full proof to make the paper self-contained and to clarify how Theorem \ref{thm:great} is used. In fact, we believe our  presentation of the proof has some merits in terms  of clarity, accessibility, and polish, at least compared to earlier papers by Litak. 

\begin{proof}
Assume $\nml{L}$ is not a join-splitting and $\nml{L}'$ is the greatest join-splitting contained in $\nml{L}$, i.e., the join of all such join-splittings. By Lemma \ref{lem:jsplfin}.\ref{jsplfin1}, $\nml{L}'$ is finitely approximable. Then since $\nml{L}'\subsetneq \nml{L}$, there is a finite Kripke frame $\ffr = \langle W, R \rangle$ for $\nml{L}'$ that refutes some $\varphi\in\nml{L}$. Furthermore, there are several assumptions we can make about $\ffr$:
\begin{itemize}[leftmargin=15mm]
\item[\mings] We can choose $\ffr$ to be rooted (i.e., there exists $r \in W$ such that every point can be reached from $r$ via the transitive closure of $R$) and such that every proper generated subframe of $\ffr$ is a frame for $\nml{L}$, by using standard preservation results and the finiteness of $\ffr$. In particular, $\varphi$ is valid over any proper generated subframe of $\ffr$. 
 \item[\mincc] In addition, using Lemma \ref{lem:jsplfin}.\ref{jsplfin2}, we can assume that $\ffr$ contains a \textit{cycle}, i.e., for some $k$ and some $w_0,\dots,w_{k-1}\in W$, we have $w_0Rw_1R\dots Rw_{k-1}Rw_0$. We can obviously choose the cycle $\{w_i\}_{i < k}$ to be minimal. In particular, either $k = 1$ (i.e., the cycle is of the form $\{w_0\}$) or $k > 1$ and $\{w_i\}_{i < k}$ contains no proper subcycles, so all points are irreflexive.
 \item[\single]  Finally, by Makinson's Theorem \citep{Makinson1971}, we also know that $\nml{L}$ is contained in either $\nml{Triv}$ or $\nml{Ver}$. 
\end{itemize}

We are going to reuse the techniques and frames used in the proof of Theorem \ref{th:ver} or Theorem \ref{th:triv}, depending on the subcase of \single\ that holds.  But first, we need to transform $\Ff$ a little bit using the above assumptions, in particular \mincc. Let $W'$ be $W$ with the elements of
  $\{w_i\}_{i < k}$ multiplied $l \deq	
  md(\varphi) + 1$ times, where $md$ is the \emph{modal degree} of a formula---the maximal number of nesting modalities. In other words,  $\{w_i\}_{i < k}$ is
  replaced by $\{w^j_i\}_{i < k,\, j < l}$. Note that $w_i$ can be identified
  with $w^0_i$; formally, we can define an embedding
  \[
  f' (w) \deq \begin{cases}
w^0_i & \text{if } w = w_i\,\text{ for some }\, i < k \\
w & \mbox{otherwise,}
\end{cases} 
  \]
  but it is convenient to suppress  the embedding in the notation. We can also define an auxiliary surjective function $f$ in the reverse direction:
\[
f(w) \deq \begin{cases}
w_i & \text{if } w = w^j_i\,\text{ for some }\, i < k, j < l, \\
w & \mbox{otherwise.}
\end{cases}
\]
Clearly, for any  $w \in \Ff$, $f(f'(w)) = w$. The accessibility relation on the extended frame is defined as follows: 
 for every  $u, v \in W'$, $u R' v$
  if either
  \begin{itemize}
  \item $v$ is not in the cycle and $f(u) R f(v)$, or
  \item $u$ is not in the cycle and for some $i < k$, $v = w^0_i$ and $f(u) R f(v)$, or
  \item for some $j
  < l$, $i < k-1$, $u = w^j_i$ and $v = w^j_{i+1}$, or
  \item for some $j
  < l$,  $u = w^j_{k-1}$ and $v = w^{(j+1) mod\ l}_0$.
  \end{itemize}
   These definitions force that $f$ is a
  bounded morphism, i.e., $\Ff$ is a bounded morphic image of $\Ff' \deq \langle W',R'\rangle$ via $f$.\footnote{Note that $f'$ is not a bounded morphism: $\Ff$ is not a \emph{generated} subframe of $\Ff'$.} 
  This, in turn, implies that 
\begin{center}
(*) for any
  valuation $V$ in $\Ff$, any $w \in \Ff$, and any $\psi$, $w \in
V(\psi)$ iff $w \in V'(\psi)$, 
\end{center}
where
$V'(p)\deq f^{-1}[V'(p)]$ for every variable $p$. Moreover, using \mincc\ we can show that (*)
holds for those $\psi$'s that are subformulas of $\varphi$ even if
$\langle \Ff' ,V'\rangle$ is replaced  by any  model based on a frame $\Gf \deq \langle U, S\rangle$
  that contains  $\Ff'$ as a (not necessarily generated) subframe in such a way that  for
every $v \in \Ff'$ with $v  \neq w^{l-1}_{k-1}$ and for every $u \in U$ with $v S u$, we have $u \in \Ff'$.  In the words of Chagrov and Zakharyaschev \citeyearpar{Chagrov1997} (adjusted to our notation), \emph{we can hook some other model} on $w^{l-1}_{k-1}$ \emph{and points in $\Ff$ will not feel its presence by means of $\varphi$'s subformulas}. Our goal in expanding the cycle beyond the modal depth of $\varphi$ was precisely to guarantee this ``insensitivity''.

Finally, let us set $t \deq |W'| + 1$. 

\begin{figure}
\begin{center}
\begin{picture}(90,40)

\put(5,8){\dashbox{1}(30,17){$\vbfr_I - \{a_0\}$}}

\put(29,19){$\bullet$} \put(29,21.4){$a_1$}
\put(29,9){$\bullet$} \put(29,11.4){$c$}

\multiput(30.8,20)(20,0){2}{\vector(1,0){8}}
\multiput(39,19)(20,0){2}{$\bullet$}
\multiput(45,20)(1.25,0){3}{\circle*{0.3}}
\put(59,19.5){\vector(-3,-1){28.2}}

\put(39,21.4){$d_t$}  \put(59,21.4){$d_0$}

\put(50,25){\dashbox{0.3}(15,15)[t]{}}
\put(57,36){$\ffr$}
\put(49,24){$\bullet$}

\put(64.8,9){$\mathbf{*\!\!\!*}$} \put(66.2,11.2){$e$}

\put(57,29){$\bullet$} \put(55.8,31.7){$w^{0}_0$}
\put(64,29){$\bullet$} \put(60.8,31.7){$w^{0}_{k-1}$}
\put(69,29){$\bullet$} \put(67.8,31.7){$w^1_0$}
\put(84,29){$\bullet$}

\put(78,30){\oval(14,20)[br]}
\put(78,20){\vector(-1,0){17}} 

\put(78,30){\oval(14,6)[br]}
\put(78,27){\line(-1,0){18}} 
\put(60,28){\oval(4,2)[bl]}
\put(58,28){\vector(0,1){1.1}}

\put(47.5,25){$r$}
\put(84,31){$w^{l-1}_{k-1}$}

\put(65.8,30){\vector(1,0){3.2}}
\multiput(60,30)(1.25,0){3}{\circle*{0.3}}
\multiput(72,30)(1.25,0){3}{\circle*{0.3}}
\multiput(78,30)(1.25,0){3}{\circle*{0.3}}

\multiput(66,40)(0.9,-0.42){20}{\circle*{0.1}}

\put(59,20.5){\vector(-2,1){8.3}}

\put(60.8,19){\vector(1,-2){4.2}} 

\put(75,37){$\ffr'-\ffr$}

\end{picture}
\vspace{0.5cm}
\caption{Frame $\ffr^\bullet_I$ (or $\ffr^\circ_I$, depending on reflexivity of $e$).  \label{figure:FI}}
\end{center}
\end{figure}

In order to proceed with the proof, we split \single\ into two cases, depending on whether or not $\nml{L}$ is contained in the logic of the single irreflexive point.

\bigskip

Case (\singirr): $\nml{L} \subseteq \nml{Ver}$. Fix $I \subseteq \omega - \{0,1\}$. Define $\ffr^\bullet_I \deq \langle W'_I, R^\bullet_I, \admis{F}'_I \rangle$, depicted in Figure \ref{figure:FI}, as follows:
$$W'_I := W' \cup \{d_0, \dots, d_t \} \cup \{e\} \cup W_I -\{a_0\},$$
where $W'$ is defined as above and $W_I$ and $a_0$ are as defined in the proof of Theorem \ref{th:ver},
\begin{align*}
R^\bullet_I := \, & R' \cup (R_I \cap W_I -\{a_0\}) \,\cup \\
& \{\langle d_{i+1},d_i \rangle\}_{i < v} \,\cup  \\
& \{\langle d_0, r\rangle\} \cup \{\langle d_0, c\rangle\} \cup \{\langle d_0, e\rangle\} \,\cup \\
& \{\langle w^{l-1}_{k-1}, d_0\rangle\}  \cup \{\langle a_1, d_t\rangle\},
\end{align*}
and $\admis{F}'_I$ consists of the sets of the form $X \cup Y$, where $X$ is an admissible subset of the frame $\vbfr_I$ defined in Theorem \ref{th:ver} and  $Y$ is a finite or cofinite subset of $W'_I - W_I$. 

We begin by defining an auxiliary sequence of formulas, which will help to name points from $\{d_t, \dots, d_0\}$: 
\begin{eqnarray*}
\delta_0(p) & \equals & \Diamond p \\
\delta_{i + 1}(p) & \equals & \Diamond\delta_i(p) \wedge \neg\delta_i(p) \wedge \bigwedge\limits_{j < i}\neg\univdmd \delta_j(p)\\
\underline{d_i} & \equals & \delta_i(\Box\bot).
\end{eqnarray*}
Note that for $i < t$, $\underline{d_i}$ may also happen to be true somewhere in $\ffr'$. However, the denotations of $\underline{d_i}$ and $\underline{d_j}$ have to be disjoint whenever $i \neq j$. By the assumption on $t$, $\underline{d_t}$ cannot be true anywhere in $\ffr'$. It is easy to verify that it cannot hold anywhere inside the $W_I-\{a_0\}$ part either, so $d_t$ is the only point in $\ffr^\bullet_I$ where $\underline{d_t}$ holds (the valuation is irrelevant, as this is a variable-free formula).

Now recall $\{\alpha_i(p)\}_{i \in \omega}$ and $\gamma(p)$ from the proof of Theorem \ref{th:ver}. We can reuse them to redefine variable-free names for points from the $W_I-\{a_0\}$ part as 
\[
 \underline{a_i}  \equals \alpha_i(\underline{d_t}) \quad\mbox{and}\quad 
 \underline{c}  \equals  \gamma(\underline{d_t}). 
\] 

Define $\nml{L}_I \deq \{ \psi \in \nml{L} \mid \ffr^\bullet_I \vDash \psi \}$.  Fix a fresh $q$ not occurring in $\varphi$, and recall that $\varphi$ can be refuted at $r$.
 Since for any $i > 1$,
\[
 \Box((\underline{a_i} \wedge \univdmd^{i + m + 1}\neg\varphi)\imply q) \vee
\Box((\underline{a_i}  \wedge \univdmd^{i + m + 1}\neg\varphi)\imply \neg q) 
\]
belongs to $\nml{L}_I$ iff $i \not\in I$, we have a continuum of distinct logics. We now want to show that whenever $\Ag \in \clofr{V}$ validates $\nml{L}_I$, it also validates $\nml{L}$.   Assume otherwise, i.e., that there is a $\psi \in \nml{L}$ such that $\Ag \nvDash \psi$. By (\singirr), we have that $\psi \in \nml{Ver}$. Hence, $\psi$ cannot be refuted at $e$. Furthermore, by the construction of  $\ffr^\bullet_I$, the point $d_0$ and hence any other point in the frame can be reached from any point refuting $\psi$. Let $l$ be a sufficiently large number; note that, say, $2 \cdot t + 4$ would do. Also, reuse the definition 
$$\dualrel\varphi \equals \Box(\univdmd\underline{a_1} \to \varphi)$$ from the proof of Theorem \ref{th:ver}. Again, fix a fresh $p$ not occurring in $\psi$.  Then the following formulas are theorems of $\nml{L}_I$:
\begin{align*}
\neg\psi  & \to \univdmd^l(\underline{c} \wedge \univdmd^l\neg\psi); \\
(\underline{c} \wedge \univdmd^l\neg\psi)  & \to   \Box(\dualrel(\dualrel p \to p) \to p)).
\end{align*}

The proof can now be completed using Theorem \ref{thm:great} in the same way as we did in the proof of Theorem \ref{th:ver}.

\newcommand{\deno}[2]{\llbracket #1 \rrbracket^\mathcal{#2}}

Case (\singref): $\nml{L} \not\subseteq \nml{Ver}$, i.e., $\nml{D} \subseteq \nml{L} \subseteq \nml{Triv}$. The only difference in the definition of $\ffr^\circ_I$ is the use of $R^\circ_I  \deq  R^\bullet_I \cup \{\langle e, e\rangle\}$. Now, just like in the proof of Theorem \ref{th:triv}, we cannot use variable-free formulas to name points. 
We begin by recalling the ``large enough'' $l$ (say, $2 \cdot t + 4$) from the previous case and letting $q, r$ be fresh for $\varphi$.  Define  
\begin{eqnarray*}
\underline{d_0}(q) & \deq  & (\univdmd^l \neg q\wedge \Diamond\univbox^l q) \vee (\univdmd^l q \wedge \Diamond\univbox^l \neg q) \\
\underline{d_{i + 1}}(q) & \equals & \Diamond\underline{d_i}(q) \wedge \neg\underline{d_i}(q) \wedge \bigwedge\limits_{j < i}\neg\univdmd \underline{d_j}(q).
\end{eqnarray*}
Now reuse the same $\alpha_i$ and $\gamma$ as above to define 
\[
 \underline{a_i}(q)  \equals \alpha_i(\underline{d_t}(q)) \quad\mbox{and}\quad 
 \underline{c}(q)  \equals  \gamma(\underline{d_t}(q)),
\] 
and just as in the proof of Theorem \ref{th:triv}, note that if
\[
\zeta(q) \deq \dile{l} q \wedge  \dile{l}\neg q
\]
is satisfied in an admissible model $\mathcal{M}$ based on $\ffr^\circ_I$, then $\underline{a_i}(q)$ does the job of a unique name, i.e., if $\deno{\zeta(q)}{M}  \neq \varnothing$, then  $\deno{\underline{a_i}(q)}{M} = \{a_i\}$.

Define $\nml{L}_I \deq \{ \psi \in \nml{L} \mid \ffr^\circ_I \vDash \psi \}$. Since for any $i > 1$,
\[
 \zeta(q) \to \univbox^l((\underline{a_i}(q) \wedge \univdmd^l\neg\varphi)\imply r) \vee
\univbox^l((\underline{a_i}(q)  \wedge \univdmd^l\neg\varphi)\imply \neg r) 
\]
belongs to $\nml{L}_I$ iff $i \not\in I$, we have a continuum of distinct logics. 

Again,  we  want to show that whenever $\Ag \in \clofr{V}$ validates $\nml{L}_I$, it also validates $\nml{L}$.   Assume otherwise, i.e., that there is a $\psi \in \nml{L}$ such that $\Ag \nvDash \psi$; using reasoning analogous to that in the (\singirr) case we get that $\psi \in \nml{Triv}$. 
 Let $var(\psi)$ be the collection of all free variables in $\psi$. The following formula holds in $\ffr^\circ_I$:
 \[
 \neg \psi \to \bigvee\limits_{q \in var(\psi)}\zeta(q).
 \]
(Assume it does not. Then there exists an admissible valuation refuting $\psi$ such that the denotations  of variables in $var(\psi)$ are either empty or equal to $W'_I$. But then this countermodel for $\psi$ can be collapsed to a single reflexive point, a contradiction). Fix a fresh $p$ and define 
$$\dualrel^q\varphi \equals \Box(\univdmd\underline{a_1}(q) \to \varphi).$$
Then the following formulas are theorems of $\nml{L}_I$:
\begin{align*}
\neg\psi \wedge \bigvee\limits_{q \in var(\psi)}\zeta(q) & \to \univdmd^l(\univdmd^l\neg\psi \wedge \bigvee\limits_{q \in var(\psi)}(\zeta(q) \wedge \underline{c}(q))); \\
\univdmd^l\neg\psi \wedge \zeta(q) \wedge \underline{c}(q)  & \to   \Box(\dualrel^q(\dualrel^q p \to p) \to p)).
\end{align*}
 
The proof can now be completed using Theorem \ref{thm:great} in the same way as we did in the proof of Theorem \ref{th:triv}.\end{proof}

\subsection{The Polymodal Case}

We stated and proved the Blok Dichotomy for unimodal logics, but this was only for technical and notational convenience. It is possible to adjust the proof to the polymodal setting by following \citealt{Litak2008}, even though Makinson's Theorem does not hold in the signature with more than one modality. As discussed in \citealt[\S~3]{Litak2008}, one needs to use the \emph{Minimal Variety Theorem} instead, which generalizes Makinson's result. We refer the reader to \citealt{Litak2008} for all the details.

\subsection{Degrees of Relative Incompleteness} \label{sec:relative}

An intriguing if somewhat esoteric question is to investigate degrees of \emph{relative} incompleteness: given classes $\mathcal{K}_1$ and $\mathcal{K}_2$ of \baos, let us say that the \textit{degree of $\mathcal{K}_1$-incompleteness relative to $\mathcal{K}_2$} of a normal modal logic $\nml{L}$ is the cardinality of the set of logics $\nml{L}'$ such that $\nml{L}'$ is $\mathcal{K}_2$-complete and $\nml{L}$ and $\nml{L}'$ are valid over exactly the same $\mathcal{K}_1$-\baos. To motivate such comparisons, recall that many of the algebraic completeness notions have frame-theoretic equivalents: $\mathcal{CA}$-completeness is equivalent to completeness with respect to \textit{normal neighborhood frames} as in \citealt{Dosen1989}; $\mathcal{AV}$-completeness is equivalent to completeness with respect to \textit{discrete frames} as in \citealt{tenCate2007}; $\mathcal{CV}$-completeness is equivalent to completeness with respect to \textit{full possibility frames} as in \citealt{Holliday2015}; and $\mathcal{V}$-completeness is equivalent to completeness with respect to \textit{principal possibility frames} as in \citealt{Holliday2015}. Degrees of relative incompleteness can be seen as providing a measure of how fine-grained these alternative modal semantics are relative to each other and to Kripke semantics.

In fact, techniques used here (and in earlier references) allow us to investigate some of these degrees of relative incompleteness. Let us briefly sketch the form of such proofs, leaving details as exercises for interested readers.  In the presence of relatively strong consequences of the Axiom of Choice such as $\nml{BPI}$, we can use the same trick as in the proofs of Theorem \ref{th:conpV} and Corollary \ref{cor:id} (and in \citealt[\S\ 4.6]{Wolter1993}  or in \citealt[\S\ 4.4]{Litak2005}) to transfer sequences of frames $\vbenthem_I$ (proof of Theorem \ref{th:ver}) and $\vbenthem^\circ_I$ (proof of Theorem \ref{th:triv}) into (general-frame duals of) neighborhood frames. In this way, we can show that for $\nml{Triv}$ and $\nml{Ver}$, their degree of \cfV-incompleteness \textit{relative to \cfC\cfA} is equal to continuum. Pushing matters further, we can similarly transform the sequences $\ffr^\bullet_I$ and $\ffr^\circ_I$ used in the proof of Theorem \ref{th:vblok} to show that a similar result obtains for any neighborhood-complete $\nml{L}$. It is somewhat more problematic to generalize this reasoning to logics that are not necessarily neighborhood-complete: one would need to use $\nml{L}^{\cfC\cfA}$, the \emph{\cfC\cfA-closure of $\nml{L}$} (the smallest neighborhood-complete logic containing $\nml{L}$), and adjust the definition of $\nml{L}_I$ to $\{ \psi \in \nml{L}^{\cfC\cfA} \mid \ffr^\bullet_I \vDash \psi \}$ or $\{ \psi \in \nml{L}^{\cfC\cfA} \mid \ffr^\circ_I \vDash \psi \}$, depending on whether (\singirr) or (\singref) holds. There seems to be, however, no guarantee that $\nml{L}^{\cfC\cfA}$ will be sound over the same class of \cfV-\baos\ as $\nml{L}$. So the strongest form of the \emph{\cfV-Dichotomy relative to \cfC\cfA} that we can show at present would be restricted to those $\nml{L}$ for which such a conservativity condition holds---in particular, neighborhood-complete ones.
 
In the reverse direction, we can show similar results about degrees of $\clofr{\omega C}$-incompleteness relative to, e.g., \cfA\cfV\ (duals of discrete frames) using techniques from \citealt{Litak2008}. The key observation is that the general frames used in the proofs of corresponding variants of the Blok Dichotomy use \emph{all} finite and cofinite sets as the collection of admissible sets; that is, they are (duals of) \cfA\cfV-\baos. Again, one can state a \emph{$\clofr{\omega C}$-Dichotomy relative to \cfA\cfV}  restricted to those $\nml{L}$ for which a corresponding \cfA\cfV-conservativity condition holds: their \emph{minimal nominal} extension  (see \S~\ref{subsec:hybrid} and \citealt{Litak2006}) is sound with respect to the same class of $\clofr{\omega C}$-\baos. In particular, this covers all \cfA\cfV-complete logics, i.e., those whose minimal nominal extension is conservative.
 
\section{Strengthening the Inference System} \label{sec:syntax}

In this and the next section, we return to the theme of \citealt{Benthem1979}: turning semantic incompleteness results into syntactic non-conservativity results. In van Benthem's example,  the formula $\Box\Diamond\top\to\Box\bot$ is not derivable from the $\vblog$-axiom according to the derivability relation $\vdash_\nml{K}^{mnu}$ of \S~\ref{sec:VB}, or equivalently, it does not belong to the normal modal logic $\vblog$, but the derivation is possible in relatively weak extensions of the logic. We can already see this abstractly: as observed in \S~\ref{subsec:RecV}, there exists a derivability relation  with a decidable notion of proof that exactly matches $\mathcal{V}$-consequence, so by \S~\ref{ssec:vb}, $\Box\Diamond\top\to\Box\bot$ is derivable from the $\vblog$-axiom in this sense. However, we would like to see non-conservativity in more concretely-given logics. As it turns out, there are several well-motivated extensions that can be used for this purpose.  

In the present section, we discuss derivations in existing calculi. The first of these, weak second-order logic (\S\ \ref{subsec:wso}), was proposed in van Benthem's \citeyearpar{Benthem1979} paper. Despite its illuminating character,  for our purposes it would be nice not to switch from the modal language to an altogether different syntax. One can derive the problematic formula in much weaker calculi that are well known in the modal community, namely, a basic nominal calculus in \S~\ref{subsec:hybrid} and the tense calculus in \S~\ref{subsec:tense}. For us,  their importance comes from the fact that they characterize consequence over classes of algebras narrower than the class of $\mathcal{V}$-\baos{}: $\mathcal{AV}$-\baos{} in the nominal case and $\mathcal{T}$-\baos{} in the tense case.

\subsection{Weak Second-Order Logic}\label{subsec:wso} 

We begin with the extension van Benthem considered: \textit{weak second-order logic}. As suggested in \S~\ref{sec:intro}, the idea is to translate formulas $\varphi$ of the modal language into formulas $SO(\varphi)$ of the monadic second-order language with a single binary relation symbol---and then deduce $SO(\Box\Diamond\top\to\Box\bot)$ from $SO(\vblog\mbox{-axiom})$ in some second-order calculus, guided by the informal proof that every Kripke frame that validates $\vblog$ validates $\Box\Diamond\top\to\Box\bot$ (Lemma \ref{lem:1}). For a $\varphi$ containing propositional variables $p_1,\dots,p_n$, $SO(\varphi)$ is defined as $\forall P_1\dots\forall P_n \forall x\, ST_x(\varphi)$ where the \textit{standard translation} $ST_x(\varphi)$ of $\varphi$ is defined recursively as usual, with the key clauses $ST_x(p_i)=P_ix$ and $ST_x(\Box\varphi)=\forall y (Rxy\rightarrow ST_y(\varphi))$ where $y$ is a fresh variable. What van Benthem calls `weak second-order logic' is the deductive system for the monadic second-order language that extends a complete axiomatic system for first-order logic with the following axioms for the monadic second-order quantifier:
\begin{itemize}
\item $\forall P(\varphi\rightarrow\psi)\rightarrow (\forall P\varphi\rightarrow \forall P\psi)$;
\item $\varphi\rightarrow \forall P\varphi$ where $P$ does not occur free in $\varphi$;
\item $\forall P\varphi\rightarrow \varphi [\psi/P]$ where $\psi$ is a first-order formula having some free variable $x$ such that $\varphi[\psi/P]$ is the result of replacing subformulas of the form $Pu$ by $\psi[u/x]$, subject to the usual qualifications about free and bound variables.
\end{itemize}
To see that $SO(\Box\Diamond\top\to\Box\bot)$ is derivable from $SO(\vblog\mbox{-axiom})$ in this weak second-order logic, first observe that $SO(\vblog\mbox{-axiom})=SO( \Box\Diamond\top\rightarrow \Box (\Box (\Box p\rightarrow p)\rightarrow p))$ is 
\begin{eqnarray*}
&&\forall P\big( \forall x(Rxy\rightarrow \exists z(Ryz\wedge \top))\rightarrow \\
&& \quad\;\;\, \forall y\big(Rxy\rightarrow \big(\forall z\big(Ryz\rightarrow \big(\forall u (Rzu\rightarrow Pu)\rightarrow Pz\big)\big)\rightarrow Py\big)\big)\big).
\end{eqnarray*}
Using the third of the second-order axioms, we can remove the $\forall P$ and substitute $v\not=y$ for $Pv$, for each variable $v$, to obtain:
\begin{eqnarray*}
&&\forall x(Rxy\rightarrow \exists z(Ryz\wedge \top))\rightarrow \\
&& \quad\;\;\, \forall y\big(Rxy\rightarrow \big(\forall z\big(Ryz\rightarrow \big(\forall u (Rzu\rightarrow u\not=y)\rightarrow z\not=y\big)\big)\rightarrow y\not=y\big)\big).
\end{eqnarray*}
Then it is straightforward to derive $SO(\Box\Diamond\top\to\Box\bot)$, which is equivalent to 
\[\forall x(\forall y(Rxy\to\exists z  Ryz) \to \forall y \neg Rxy),\] 
by a formalized version of the proof of Lemma \ref{lem:1}. Thus, although the modal relation $\vdash_\nml{K}^{mnu}$ of \S~\ref{sec:VB} is too weak to derive $\Box\Diamond\top\to\Box\bot$ from the $\vblog$-axiom, the weakest reasonable system of second-order logic is enough to do so under translation.

 As van Benthem \citeyearpar{Benthem1978} observed, the relation $\vdash_{wso}$ of derivability in weak second-order logic axiomatizes a natural notion of semantic consequence. Interpret the second-order language with a binary relation symbol in general frames $\gfr=\langle W,R, \admis{W}\rangle$ where $\langle W,R\rangle$ is a Kripke frame and $\admis{W}\subseteq\wp(W)$ is closed under first-order definability: given any formula $\varphi$ of the second-order language whose free individual variables are $x,x_1,\dots,x_n$, whose free predicate variables are $X_1,\dots,X_m$, and which does not contain any second-order quantifiers, if $w_1,\dots,w_n\in W$ and $A_1,\dots,A_m\in\admis{W}$, then $$\{w\in W\mid \gfr \vDash \varphi [w,w_1,\dots,w_n,A_1,\dots,A_m]\}\in\admis{W}.$$ Define the consequence relation $\vDash_{wso}$ by: $\Sigma \vDash_{wso} \varphi$ iff for all general frames $\gfr$ as above and all variables assignments $f$ mapping individual variables to elements of $W$ and predicate variables to elements of $\admis{W}$, if $\gfr\vDash \sigma [f]$ for every $\sigma\in\Sigma$, then $\gfr\vDash \varphi[f]$. Then one can show the following.

\begin{proposition}\label{prop:comp1} $\Sigma \vdash_{wso}\varphi$ iff $\Sigma\vDash_{wso}\varphi$.
\end{proposition}

\noindent Thus, like $\mathcal{V}$-consequence (recall \S~\ref{subsec:RecV}), the relation $\vDash_{wso}$ is recursively axiomatizable.

It is noteworthy that the \baos\ underlying the general frames for weak second-order logic above are $\mathcal{AT}$-\baos. The $\mathcal{A}$ part is obvious, since $\admis{W}$ must contain all singleton subsets of $W$ by the requirement of closure under definability. For the $\mathcal{T}$ part, where the \bao\ operator $\Diamond$ is given by $\Diamond A= R^{-1}[A]$ for $A\in\admis{W}$, its residual $\resdual$ is given by $\resdual A = \{w\in W\mid R^{-1}(w)\subseteq A\}$. This is clearly definable as above, and it is easy to see that for any $A,B\in \admis{W}$, we have $\Diamond A\subseteq B$ iff $A\subseteq \resdual B$. This shows that the underlying \bao\ is a $\mathcal{T}$-\bao\ (and hence a $\mathcal{V}$-\bao). Thus, for any $\sigma$ and $\varphi$ such that  $\sigma \vDash_{\cfA\cfT} \varphi$, we also have $SO(\sigma)\vDash_{wso} SO(\varphi)$ and hence $SO(\sigma)\vdash_{wso} SO(\varphi)$ by Proposition \ref{prop:comp1}.  

\subsection{Minimal Nominal Extension}\label{subsec:hybrid}

Let $\mathcal{H}$ be a propositional modal language with two types of atomic formulas: \textit{propositional variables} $p$, $q$, $r$ \dots, and \textit{nominals} $i$, $j$, $k$ \dots\ (cf.~Footnote \ref{f:firsthyb}). In the intended Kripke semantics, the difference between propositional variables and nominals is that nominals must be evaluated as singleton sets instead of arbitrary subsets. 

A \textit{nominal modal logic} is a set $\nml{L}$ of formulas of the language $\mathcal{H}$ that satisfies the conditions of a normal modal logic (with uniform substitution of formulas for propositional variables, but not for nominals) and the following:
\begin{itemize}
\item[(a)] $\nml{L}$ is closed under uniform substitution of nominals for nominals;
\item[(b)] for every $n\in\mathbb{N}$ and nominal $i$, $\Box^{\leq n}(i\rightarrow\varphi)\vee \Box^{\leq n}(i\rightarrow\neg\varphi)\in \nml{L}$, where $\Box^{\leq n}\varphi~:=~\varphi\wedge\Box\varphi\wedge\dots\wedge \Box^n\varphi$;
\item[(c)] if $\ell(\neg i)\in \nml{L}$ for all nominals $i$, then $\ell(\bot)\in\nml{L}$,
\end{itemize}
where $l$ is any \textit{necessity form}, defined as follows. Fixing a symbol $\natural$ not occurring in formulas of $\mathcal{H}$, the set of necessity forms is the smallest set containing $\natural$ such that for all necessity forms $\ell(\natural)$ and $\varphi\in\mathcal{H}$, both $(\varphi\rightarrow \ell(\natural))$ and $\Box l(\natural)$ are also necessity forms.\footnote{\label{ft:principal}The leading occurrence of $\Box$ in the $\Box \ell(\natural)$ clause and in all necessity forms built using $\Box \ell(\natural)$ will be called a \emph{principal} one. This will matter when considering polymodal generalizations, e.g., in \S~\ref{subsec:PureModal}.} In the statement of condition (c), $\ell(\varphi)$ is the formula of $\mathcal{H}$ obtained by substituting $\varphi$ for $\natural$ in $\ell(\natural)$. The condition (c) is the COV rule of \citealt{Gargov1993}. 

Let the \textit{minimal nominal extension} $\nml{L}.n$ of a normal unimodal logic $\nml{L}$ be the smallest nominal modal logic that includes $\nml{L}$. Ten Cate and Litak \citeyearpar{tenCate2007} showed that conservativity of $\nml{L}.n$ over $\nml{L}$ is equivalent to $\nml{L}$ being the logic of some class of \textit{discrete} general frames---general frames in which every singleton subset is admissible---and that discrete frames are duals of $\mathcal{AV}$-\baos, which yields the following.

\begin{proposition}[Ten Cate and Litak]\label{prop:discrete} A normal modal logic $\nml{L}$ is $\mathcal{AV}$-complete iff $\nml{L}.n$ is a conservative extension of $\nml{L}$.
\end{proposition} 

Since the logic $\vblog$ is $\mathcal{V}$-incomplete and hence $\mathcal{AV}$-incomplete, it follows from Proposition \ref{prop:discrete} that its minimal nominal extension $\vblog.n$ is not a conservative extension. Indeed, it is easy to show that $\Box\Diamond\top\to\Box\bot$ belongs to $\vblog.n$. First observe that by condition (b) above, the smallest nominal modal logic contains $i \to \Box(i \to \Diamond i)$ and hence
\[
\Box(i \to \Box(i \to \Diamond i)).
\]
Then as an instance of the $\vblog$-axiom with the conditional in the consequent contraposed, we have
\[
\Box\Diamond\top \to \Box(i \to \Diamond(i \wedge \Box \neg i))\in \vblog.n.
\]
Using the normality of $\Box$, the two formulas above yield:

\[
\Box\Diamond\top \to \Box \neg i \in\vblog.n,
\]
whence $\Box\Diamond\top\to\Box\bot\in \vblog.n$ by the COV rule.\\

\subsection{Minimal Tense Extension}\label{subsec:tense}

The \textit{minimal tense extension} $\nml{L}.t$ of a normal unimodal logic $\nml{L}$ is the smallest normal bimodal logic including $\nml{L}$ and the axioms $p\rightarrow \Box\Diamond^{-1}p$ and $p\rightarrow \Box^{-1}\Diamond p$, where $\Diamond^{-1}$ and $\Box^{-1}$ are the new pair of modal operators.

The following analogue of Proposition \ref{prop:discrete} for $\mathcal{T}$-\baos\ is easy to see, given that the Lindenbaum-Tarski algebra of $\nml{L}.t$ is always a $\mathcal{T}$-\bao.

\begin{proposition}\label{prop:Tconserv} A normal modal logic $\nml{L}$ is $\mathcal{T}$-complete iff $\nml{L}.t$ is a conservative extension of $\nml{L}$.
\end{proposition}

Since the logic $\vblog$ is $\mathcal{V}$-incomplete and hence $\mathcal{T}$-incomplete, it follows from Proposition \ref{prop:Tconserv} that its minimal tense extension $\vblog.t$ is not a conservative extension. That a tense extension may fail to be conservative is fascinating (see \citealt[\S5.4]{Wolter1993}; \citealt[\S3.3]{Kracht1997}; \citealt[p.~170]{Goldblatt2001}), especially when one looks at a concrete derivation. We will sketch such a derivation, using the metatheorem for all tense logics that $\Box$ and $\Diamond^{-1}$ are residuals, so $\varphi\rightarrow \Box\psi$ is a theorem iff $\Diamond^{-1}\varphi\rightarrow\psi$ is a theorem. Where $\chi := \Diamond^{-1}\Box\Diamond\top$, we have the following theorems of $\vblog.t$:
\begin{enumerate}
\item[0.] $\Box\Diamond\top\rightarrow \Box(\Box(\Box\neg \chi\rightarrow \neg \chi)\rightarrow\neg \chi)$ \quad instance of $\vblog$-axiom
\item[1.] $\Box\Diamond\top\rightarrow \Box\Diamond^{-1}\Box\Diamond\top$ \quad instance of tense axiom
\item[2.] $\Box\Diamond\top\rightarrow\Box \chi$ \quad from 1 by definition of $\chi$
\item[3.] $\Box\Diamond\top\rightarrow \Box (\neg\Box (\Box\neg \chi\rightarrow \neg \chi))$ \quad from 0 and 2 by normal modal reasoning
\item[4.] $\Diamond^{-1}\Box\Diamond\top\rightarrow \neg\Box (\Box\neg \chi\rightarrow \neg \chi)$ \quad from 3 by residuation
\item[5.] $\chi\rightarrow \neg\Box (\Box\neg \chi\rightarrow \neg \chi)$ \quad from 4 by definition of $\chi$
\item[6.] $\chi\rightarrow \Diamond (\chi\wedge\Box\neg \chi)$ \quad rewriting 5
\item[7.] $\chi\rightarrow \Diamond (\Diamond (\chi\wedge\Box\neg \chi)\wedge\Box\neg \chi)$ \quad from 6 by normal modal reasoning
\item[8.] $\chi\rightarrow \bot$ \quad from 7 by normal modal reasoning
\item[9.] $\Diamond^{-1}\Box\Diamond\top\rightarrow\bot$ \quad from 8 by definition of $\chi$
\item[10.] $\Box\Diamond\top\rightarrow\Box\bot$ \quad from 9 by residuation.
\end{enumerate}
But $\Box\Diamond\top\rightarrow\Box\bot\not\in\vblog$ by Lemma \ref{lem:2}, so $\vblog.t$ is not a conservative extension.\footnote{An intriguing fact observed by Humberstone \citeyearpar[p.~870]{Humberstone2011} is that $\Box\Diamond\top\rightarrow\Box\bot$ can be derived from the  $\vblog$-axiom using the so-called Halld\'{e}n necessitation rule.}

\section{Toward a Syntactic Characterization of $\mathcal{V}$-Consequence} \label{sec:newsyntax}

\newcommand{\lgqm}{\ensuremath{\mathcal{L}_{\mathrm{GQM}}}}
\newcommand{\lubox}{\ensuremath{\mathcal{L}(\ubox)}}

In \S~\ref{sec:syntax}, we saw how existing calculi and rules characterizing consequence over narrower classes of algebras allow us to show the incompleteness of $\vblog$. In this section, we are going to discuss new extensions of the deductive apparatus that are sound for $\mathcal{V}$-consequence and also allow us to prove the incompleteness of $\vblog$ and $\nml{GLB}$. Our rule(s) are inspired by the reformulation of complete additivity as condition $\mathcal{R}$ in \S~\ref{sec:R&V}. In particular, it will be convenient to work with the following obviously equivalent version of $\mathcal{R}$.

\begin{lemma}\label{lem:Rref} The condition $\mathcal{R}$ is equivalent to: for all $a,b\in\Ag$, \textit{if} for all $c\in\mathfrak{A}$ such that $b\leq c < \top$, there is a $d\in\mathfrak{A}$ such that $c\leq d < \top$ and $a\leq\Box d$, \textit{then} $a\leq\Box b$.
\end{lemma}

This \textit{if} \dots\ \textit{then} has the feel of an inference rule that we could try to write in a modal language. We will start by showing how this can be done in a rather expressive modal language. Then we will gradually limit expressivity, finally achieving rule(s) expressible in the plain modal syntax, yet still sufficient for incompleteness proofs.

\subsection{Extended Languages}\label{subsec:ExtLang}

The first language we will consider is the language \lgqm\ of \textit{Global Quantificational Modalities} \citep{H&L2016}, given by the following grammar:
\[\varphi ::= p\mid \neg\varphi \mid (\varphi\wedge\varphi) \mid \Box \varphi\mid [\forall p]\varphi\mid [\exists p]\varphi.\]

Before we give a formal semantics, or even some intuition for these  quantificational modalities, let us introduce the other language we are concerned with in this subsection:  \lubox, the extension of the basic unimodal language with the universal modality $\ubox$. The dual of $\ubox$, the existential modality $\udiam$, is defined as usual by $\udiam \varphi:=\neg\ubox\neg\varphi$.  For a normal unimodal logic $\nml{L}$, let $\nml{L}.\mathsf{A}$ be the smallest bimodal logic---with modalities $\Box$ and $\ubox$---that extends $\nml{L}$ with the  $\nml{S5}$ axioms for $\ubox$ and the axiom $\ubox p\rightarrow \Box p$. Note that $\nml{L}.\mathsf{A}$ is always a conservative extension of $\nml{L}$: a modal algebra validating $\nml{L}$ becomes an algebra validating $\nml{L}.\mathsf{A}$ with the interpretation $\hat{\theta}(\mathsf{A}\varphi)$ as $\top$ if $\hat{\theta}(\varphi)=\top$ and $\bot$ otherwise.

Let us now return to \lgqm. The intended semantics for $[\forall p]$ is the same as for $\forall p\ubox \varphi$. The intended semantics for $[\exists p]\varphi$ is the same as for $\exists p\ubox \varphi$ (note that we are using the universal modality here as well, not its dual!). In a \bao, we have the following interpretations:
\begin{eqnarray*}
\hat{\theta}([\forall p]\varphi)&=&\begin{cases} \top &\mbox{if }\widehat{\theta'}(\varphi)=\top\mbox{ for all }\theta'\mbox{ that differ from }\theta\mbox{ at most at }p \\ \bot &\mbox{otherwise}\end{cases} \\
\hat{\theta}([\exists p]\varphi)&=&\begin{cases} \top &\mbox{if }\widehat{\theta'}(\varphi)=\top\mbox{ for some }\theta'\mbox{ that differs from }\theta\mbox{ at most at }p \\ \bot &\mbox{otherwise.}\end{cases}.
\end{eqnarray*}

\lgqm\ is clearly more expressive than $\mathcal{L}(\ubox)$: the universal modality $\ubox$ can be defined by $\ubox \varphi:= [\forall p] \varphi$ where $p$ does not occur in $\varphi$. But just like \lubox, \lgqm\ can still be interpreted in any algebra. 
This distinguishes \lgqm\ from typical modal languages involving propositional quantifiers: such quantifiers are normally interpreted using infinite operations, which
poses problems in the absence of lattice-completeness. For more on \lgqm, its semantics, axiomatization, and expressive power, see \citealt{H&L2016}.

The version of $\mathcal{R}$ in Lemma \ref{lem:Rref} can be directly translated into a sentence of \lgqm:
$$
[\forall p][\exists q]\big(\big(\ubox(\beta \to p) \wedge \udiam \neg p\big) \to \big(\ubox(p \to q) \wedge \udiam \neg q \wedge \ubox(\alpha \to \Box q)\big)\big) \to \ubox(\alpha \to \Box\beta).
$$
In order to arrive at a principle expressible in \lubox, let us first weaken the above axiom to a rule:
$$
\inferrule{[\forall p] [\exists q] \big (\big(\ubox(\beta \to p) \wedge \udiam \neg p\big) \to \big(\ubox(p \to q) \wedge \udiam \neg q \wedge \ubox(\alpha \to \Box q)\big)\big)}{\alpha \to \Box\beta}.
$$
 While the rule is easy to understand semantically, it is not very convenient to work with in deductions. One would need to use contraposition with almost every conceivable implication. So there is another version of the rule, where $\udiam$ and $\neg$ do not occur, yielding something more interesting from a constructive point of view:

$$
\inferrule{[\forall p] [\exists q]\big(\ubox(\beta \to p) \to \big(\ubox p \vee \big(\ubox(p \to q) \wedge (\ubox q \to \ubox p)\wedge\ubox(\alpha \to \Box q)\big)\big)\big)}{\alpha \to \Box\beta}.
$$

\newcommand\new{\reflectbox{\ensuremath{\mathsf{N}}}}
\newcommand\hnu{\scalebox{.7}{\new}}

 A natural strategy to obtain a corresponding rule in  pure \lubox\ is to use suitable 
\emph{freshness assumptions}, which we will write using nominal-logic-like notation \citep{Pitts2013,Pitts2016}  (cf. Footnote \ref{f:firsthyb}), with `$p\# \alpha,\beta$' ($p$ is for \textit{fresh for} $\alpha$, $\beta$) meaning that $p$ does not occur in $\alpha,\beta$:
\begin{equation}
\inferrule{\big(\ubox(\beta \to p) \wedge \udiam \neg p\big) \to \big(\ubox(p \to \chi(p)) \wedge \udiam \neg \chi(p) \wedge \ubox(\alpha \to \Box \chi(p))\big) \quad p \# \alpha, \beta}{\alpha \to \Box\beta}.\tag{$\mathcal{V}$-inf}\label{V-inf}
\end{equation}
By writing `$\chi(p)$', we are only stressing that $p$ \emph{can} occur in $\chi$ (unlike in $\alpha$ or $\beta$), not that it must. Let us show that \ref{V-inf} preserves validity over any $\mathcal{V}$-\bao{}.

\begin{proposition}\label{prop:V-infSound} The inference rule \ref{V-inf} preserves validity over any $\mathcal{V}$-\bao .
\end{proposition}
 
\begin{proof} Given a $\mathcal{V}$-\bao\ $\mathfrak{A}$, suppose $\mathfrak{A}$ validates the premise of the rule, and consider any valuation $\theta$ for $\mathfrak{A}$. We claim that $\hat{\theta}(\alpha\rightarrow\Box\beta)=\top$.  Take any $c\in\mathfrak{A}$ such that $\hat{\theta}(\beta)\leq c<\top$. Let $\theta'$ be the valuation that differs from $\theta$ at most at $p$ such that $\theta'(p)= c$. Since $p$ does not occur in $\beta$, $\hat{\theta}(\beta)=\widehat{\theta'}(\beta)$, so $\widehat{\theta'}(\beta)\leq c <\top$. Then by the assumption that the premise of the rule is valid, we have $c\leq \widehat{\theta'}(\chi(p))<\top$ and $\widehat{\theta'}(\alpha)\leq \widehat{\theta'}(\Box\chi(p))=\Box \widehat{\theta'}(\chi(p))$. Let $d:=\widehat{\theta'}(\chi(p))$. Since $p$ does not occur in $\alpha$, $\widehat{\theta'}(\alpha)=\hat{\theta}(\alpha)\leq \Box d$. So we have shown that for any $c\in\mathfrak{A}$ such that $\hat{\theta}(\beta)\leq c<\top$, there is a $d\in\mathfrak{A}$ such that $c\leq d < \top$ and $\hat{\theta}(\alpha)\leq \Box d$. Then by the reformulation of $\mathcal{R}$ in Lemma \ref{lem:Rref}, $\hat{\theta}(\alpha)\leq \Box\hat{\theta}(\beta)=\hat{\theta}(\Box\beta)$, so $\hat{\theta}(\alpha\rightarrow\Box\beta)=\top$, as claimed. Thus, the conclusion of the rule is valid over $\mathfrak{A}$.\end{proof}

Proposition \ref{prop:V-infSound} shows that if a normal modal logic $\nml{L}$ is $\mathcal{V}$-complete, then the minimal normal extension of $\nml{L}.\mathsf{A}$ that is closed under the rule \ref{V-inf} is a conservative extension of $\nml{L}$. By contraposition, to show the failure of $\mathcal{V}$-completeness, it suffices to show the failure of conservativity.

As above, we have two subtly differing syntactic variants. The statement of \ref{V-inf} above is the semantically convenient version of the rule, whereas deductively (and constructively) one may have some preference for
\begin{equation*}
\inferrule{ \ubox(\beta \to p) \to \big( \ubox p \vee \big(\ubox(p \to \chi(p)) \wedge (\ubox \chi(p) \to \ubox p)\wedge\ubox(\alpha \to \Box \chi(p))\big)\big) \quad p \# \alpha, \beta}{\alpha \to \Box\beta}.
\end{equation*}

However, the premise of this variant still does not look particularly appealing. In fact, for our purposes we can use a simpler  special instance:
\begin{equation}
\inferrule{ p \to \chi(p) \\  \ubox \chi(p) \to \ubox p \\ \alpha \to \Box \chi(p) \\ p \# \alpha}{\alpha \to \Box\bot}.\tag{$\mathcal{V}$-spec}\label{V-spec}
\end{equation}
If the premises of \ref{V-spec} are valid over a \bao\ $\mathfrak{A}$, then the premise of \ref{V-inf} is also valid over $\mathfrak{A}$ with $\beta:=\bot$; so if $\mathfrak{A}$ is a $\mathcal{V}$-\bao, then the conclusion $\alpha \to \Box\bot$ is valid over $\mathfrak{A}$.

Let us use \ref{V-spec} to derive $\Box\Diamond\top \to \Box\bot$ from the $\vblog$-axiom $$\Box\Diamond\top \to {\Box(\Box(\Box p \to p) \to p)}.$$ We take $\alpha := \Box\Diamond\top$ and  choose $\chi(p):=\Box(\Box p \to p) \to p$. Our rule says that to derive $\alpha \to \Box\bot$, we need to show that the following three formulas are theorems of $\vblog.\ubox$:
\begin{center}
(a) $p \to (\Box(\Box p \to p) \to p)$ \qquad (b) $\ubox (\Box(\Box p \to p) \to p) \to \ubox p \qquad$ (c) $\vblog$.
\end{center}
 
Premise (c) is an axiom. Premise (a) is a theorem of $\nml{K}$ and thus a fortiori of $\vblog.\ubox$.  Premise (b) is a theorem of $\nml{K}.\ubox$. Its derivation is the only place where we use axioms for the universal modality. For we have
$$
\vdash_{\mathsf{K}.\ubox}  \ubox (\Box(\Box p \to p) \to p) \to \ubox (\Box p \to p),
$$
which with the theorem $\ubox q\rightarrow \ubox \Box q$ gives us
$$
\vdash_{\mathsf{K}.\ubox}  \ubox (\Box(\Box p \to p) \to p) \to \ubox \Box(\Box p \to p),
$$
which with the axiom $\ubox(q\to r)\to (\ubox q\to \ubox r)$ yields the desired
$$
\vdash_{\mathsf{K}.\ubox}  \ubox (\Box(\Box p \to p) \to p) \to \ubox p.
$$

We can go even further: there is a still more special instance of our rule, which does not use the universal modality at all and yet is sufficient both for the $\vblog$ deduction above and for the $\nml{GLB}$ incompleteness result.

\subsection{Pure Modal Syntax}\label{subsec:PureModal}

To eliminate the universal modality from \ref{V-spec}, we can replace $\ubox \chi(p) \to \ubox p$ by any stronger formula, i.e., any formula whose validity over a \bao\ implies that of $\ubox \chi(p) \to \ubox p$; the resulting rule will still be $\mathcal{V}$-sound. As we now want a rule that will work for $\nml{GLB}$ as well, let us formulate it in a polymodal syntax. The first rule  \ref{V-mod} that we present is an instance of a  \cfV-sound rule scheme \ref{Vl-mod} that we will see at the end of this section; however, we do not need full generality to cover $\vblog$ and $\nml{GLB}$ at the same time.

 Recall again that $\boxdot_i\alpha := \alpha \wedge \Box_i\alpha$. Our example of a \cfV-sound rule without the universal modality is:
\begin{equation}
\inferrule{ p \to \chi(p) \\  \boxdot_i \chi(p) \to  p \\ \alpha \to \Box_j \chi(p) \\ p \# \alpha}{\alpha \to \Box_j\bot}.\tag{$\mathcal{V}$-mod}\label{V-mod}
\end{equation}
In the unimodal case, take $i=j$. This yields a purely unimodal rule that preserves validity over any $\mathcal{V}$-\bao{} and suffices to carry out van Benthem's deduction. We leave adjusting the deduction from \S~\ref{subsec:ExtLang} so that \ref{V-mod} replaces \ref{V-spec} as an exercise.

\newcommand{\thof}[1]{\vdash_{\nml{#1}}}

It seems more interesting to note that \ref{V-mod} can be used to translate the algebraic reasoning in the proof of Theorem \ref{th:glbinc} for $\nml{GLB}$, in order to derive $\onbx\bot$. Let $i := 0$, $j := 1$, and $\alpha := \top$, so the conclusion of \ref{V-mod} becomes $\onbx \bot$. As in the derivation for $\vblog$ in \S~\ref{subsec:ExtLang}, set $\chi(p) :=  \zebx(\zebx p \to p) \to p$. Then since $\alpha:=\top$, the third premise of \ref{V-mod} becomes 
\[\onbx (\zebx(\zebx p \to p) \to p),\]
which is a bimodal version of the $\vblog$-axiom. Let us show that this is indeed a theorem of $\nml{GLB}$. First, note the well-known fact that $\thof{GLB} \onbx (\zebx p \to p)$ (In order to derive this, start with $\thof{GLB}\onbx\zedm \neg p \vee \langle 1\rangle\zebx p$ by excluded middle, which implies $\thof{GLB} \onbx\zedm \neg p \vee \zebx p $ and then $\thof{GLB} \onbx\zedm \neg p\vee \onbx p$ by axioms (iii) and (ii) of $\nml{GLB}$, which in turn implies $\thof{GLB} \onbx(\zedm \neg p\vee p)$.) Then we have:
\begin{eqnarray*}
&&\thof{GLB} \zebx (\zebx p \to p) \to \zebx p \quad \mbox{L\"{ob} axiom} \\
&& \thof{GLB} \onbx (\zebx (\zebx p \to p) \to \zebx p)\quad \mbox{by Necessitation}\\
&& \thof{GLB} \onbx (\zebx p \to p)\quad\mbox{from above} \\
&& \thof{GLB} \onbx (\zebx(\zebx p \to p) \to p)\quad\mbox{from previous two steps using normality of $\onbx$.}
\end{eqnarray*}

Turning to the first and second premises of \ref{V-mod}, as in the case of $\vblog$, these can be shown to be theorems of unimodal $\nml{K}$, i.e., in the language with $\zebx$ only. As the $\nml{K}$-theoremhood of the second premise may be  less immediate to see, let us sketch this derivation (returning to the unimodal syntax). We want to show
$$
\thof{K} \boxdot(\Box(\Box p \to p) \to p) \to p.
$$
First, one can easily derive
$$
\thof{K} (\Box(\Box p \to p) \to p) \to (\Box p \to p).
$$
By normality, this also yields
$$
\thof{K} \Box(\Box(\Box p \to p) \to p) \to \Box(\Box p \to p).
$$
Now it is enough to recall the definition of $\boxdot$ to obtain
$$
\thof{K} \boxdot(\Box(\Box p \to p) \to p) \to \big(\Box(\Box p \to p) \wedge (\Box(\Box p \to p) \to p)\big)
$$
and finish the derivation. Thus, we have shown that all three premises of \ref{V-mod} are theorems of $\nml{GLB}$. So \ref{V-mod} allows us to derive the $\mathcal{V}$-consequence $\onbx \bot$ of $\nml{GLB}$.
 
Since \ref{V-mod} is $\mathcal{V}$-sound and hence $\mathcal{T}$-sound, it is admissible in tense logics (recall \S~\ref{subsec:tense}). We will conclude this section by showing syntactically that \ref{V-mod} is admissible in tense logics with two future box modalities $\Box_i$ and $\Box_j$ (which may be the same) as in the statement of \ref{V-mod}. Substitute $\Box_j^{-1}\neg\alpha$ for $p$ in the premises of \ref{V-mod}. Then we deduce the conclusion of \ref{V-mod} as follows: 
\begin{enumerate}
\item $\boxdot_i \chi(\Box_j^{-1}\neg\alpha) \to  \Box_j^{-1}\neg\alpha$ \quad second premise of \ref{V-mod}
\item $\Diamond_j\boxdot_i \chi(\Box^{-1}_j\neg\alpha) \to \neg\alpha$ \quad from 1 by residuation 
\item $\alpha\to \Box_j\neg\big(\chi(\Box^{-1}_j\neg\alpha)\wedge \Box_i \chi(\Box^{-1}_j\neg\alpha)\big)$ \quad from 2 by contraposition,  definition of $\boxdot$
\item $\alpha\to \Box_j \chi(\Box^{-1}_j\neg\alpha)$ \quad third premise of \ref{V-mod} 
\item $\alpha\to \Box_j \Diamond_i \neg \chi(\Box^{-1}_i\neg\alpha)$ \quad from 3 and 4 
\item $\Diamond^{-1}_j\alpha \to \chi(\Box^{-1}_j\neg\alpha)$ \quad from 4 by residuation
\item $\neg\chi(\Box^{-1}_j\neg\alpha) \to \Diamond^{-1}_j\alpha$ \quad contrapositive of first premise of  \ref{V-mod}
\item $\neg\chi(\Box^{-1}_j\neg\alpha)\to\bot$ \quad from 6 and 7 
\item $\Diamond_i \neg\chi(\Box^{-1}_j\neg\alpha)\to \bot$ \quad from 8
\item $\alpha\to \Box_j\bot$ \quad from 5 and 9.
\end{enumerate}

The rule \ref{V-mod} seems too specialized to cherish hopes that it might yield the ultimate syntactic characterization of $\clofr{V}$-completeness via its conservativity.
At least, however, its conservativity over the set of theorems of a given modal logic is a \textit{necessary} condition for the $\clofr{V}$-completeness of the logic.
 Moreover, as mentioned above, it can be significantly generalized. Let $\ell(\natural), \ell_1(\natural), \dots, \ell_n(\natural)$ be necessity forms as defined in \S~\ref{subsec:hybrid}, suitably adjusted to the polymodal language. Furthermore, assume $\ell(\natural)$ does not involve $p$. In the polymodal setting, we require that if only some of the diamonds are completely additive, then only the corresponding boxes can occur in $\ell(\natural)$ in principal positions (recall Footnote \ref{ft:principal}).  Consider the following rule scheme:\footnote{For the second premise, recall that implication associates to the right.} 
 
 \begin{equation}
\inferrule{ p \to \chi(p) \\  \ell_1(\chi(p)) \to \dots \to \ell_n(\chi(p)) \to  p \\ \ell(\chi(p)) \\ p \# \ell}{\ell(\bot)}.\tag{$\mathcal{V}_{\ell}$-mod}\label{Vl-mod}
\end{equation}
To see \ref{V-mod} as an instance of \ref{Vl-mod}, set $\ell_1(\natural) \deq \natural$, $\ell_2(\natural) \deq \Box_i\natural$, and $\ell(\natural) \deq \alpha \to \Box_j\natural$. To justify \ref{Vl-mod} as a corollary of \ref{V-spec} in the same way as we justified \ref{V-mod} at the beginning of this subsection, note that $\thof{\mathsf{K}.\ubox} \ubox p \to \ell_i(p)$ and that the assumption that \cfV-boxes appear principally in $\ell(\natural)$ turns it into a term-definable \cfV-box.

We leave the question of whether conservativity of a rule scheme like \ref{Vl-mod} can provide a necessary-and-sufficient 
 characterization of \cfV-completeness 
  for future work.

\section{Conclusions and Future Work}\label{sec:conclude}

We have shown the existence of \cfV-incomplete  logics (in the unimodal signature) and \cfV-inconsistent logics (in a polymodal signature), thus answering a long-standing open question (\citealt{Litak2004}, \citealt[Ch.~9]{Litak2005b}, \citealt[\S~6.1]{Venema2007}, \citealt[\S~7]{Litak2008}). These results contrast with the fact that for many natural logics, even their countably generated free algebra can be completely additive (cf.~Footnote \ref{fn:FreeAlgebras}). Moreover, the examples involved turn out to be surprisingly natural. One of them is the logic $\vblog$ designed by van Benthem \citeyearpar{Benthem1979} to extract the syntactic essence of incompleteness results. Still more strikingly, another example is provided by the bimodal provability logic $\nml{GLB}$ \citep{Japaridze1988,Boolos1993,Beklemishev2010,Beklemishev2011}. These $\mathcal{V}$-incompleteness results essentially relied on a reformulation of \cfV\ as a first-order $\forall\exists\forall$-property inspired by the work of the first author on modal \emph{possibility frames}, which also allows a concrete description of categories dual to categories of \cfV-\baos\ \citep{Holliday2015}.

The $\vblog$ axiom and the general frame used by van Benthem \citeyearpar{Benthem1979} have often been reused in proofs of the Blok Dichotomy for various generalizations of Kripke completeness \citep{Chagrov1997,Zak2001,Litak2005b,RautenbergWZ06,WolterZ06,Litak2008}. We have shown that this strategy extends smoothly to degrees of \cfV-incompleteness. 

We were also able to follow a less frequented path opened up by van Benthem \citeyearpar{Benthem1979} and investigate syntactic aspects of incompleteness proofs. Apart from several extended modal formalisms in which to internalize incompleteness arguments, from the nominal and tense formalisms to \lgqm\ (Global Quantificational Modalities, investigated further in \citealt{H&L2016}), we have presented a single, surprisingly simple non-standard rule \ref{V-mod} in a pure modal language (generalizing further to a rule scheme \ref{Vl-mod}), which is admissible over \cfV-\baos\  and can be used to show the $\mathcal{V}$-incompleteness of both the van Benthem logic and $\nml{GLB}$.

These results raise a number of further questions. First, we still need to provide a definite characterization of \cfV-completeness in terms of conservativity of suitable \emph{minimal extensions}, similar to those available in the earlier literature for \cfA\cfV-completeness  or \cfT -completeness (nominal and tense extensions, respectively), preferably in a language not involving quantifiers or global quantificational modalities. 

Second, while we are very satisfied with $\nml{GLB}$  as a natural example of \cfV-incompleteness, we would like to see an equally striking case of \cfV-inconsistency. Recall that one of our present examples of \cfV-inconsistency is $\nml{GLBe}$, which can be naturally interpreted using ordinal semantics of polymodal provability logics \citep{Beklemishev2010,Beklemishev2011}. Can this logic or an extension thereof be given a good provability interpretation? Or is there another route to explore?

Third, while we believe that the syntactic ``internalization'' of incompleteness proofs we proposed---i.e., investigating relevant derivations in (extended) modal formalisms---is a fruitful and natural perspective, it could and perhaps should be complemented with an approach more directly following that of van Benthem \citeyearpar{Benthem1979}. Recall that van Benthem proposed using \emph{weak second-order logic} as the ambient formalism. In our setting, this would amount to characterizing \cfV-, \cfA\cfV-, \cfT-consequence and their relatives by allowing very limited instances of the axioms of weak second-order logic (e.g., involving equalities, relational atoms with variables permuted, etc.) plus possibly some additional axioms/rules. Matching this with the ``internalized'' approach above would lead to a fine-grained perspective on (extended) correspondence theory.

Fourth, when briefly discussing the subject of \emph{degress of relative incompleteness} in \S~\ref{sec:relative}, we recalled the fact that many of the algebraic completeness notions have frame-theoretic equivalents, involving  normal neighborhood frames \citep{Dosen1989}, discrete frames \citep{tenCate2007}, or possibility frames \citep{Holliday2015}. Are there other ways of systematically comparing these semantics and their associated completeness notions, superior to studying degrees of relative incompleteness?\footnote{On a side note, \emph{coalgebraic logic} (see \citealt{CirsteaEA11,SchroederPattinson10,Kurz2012} for references)  provides a generalization of Kripke semantics in a direction orthogonal to that of possibility frames: on the dual side, instead of keeping \cfV\ and dropping \cfA, one is keeping \cfC\cfA\ and does not insist on \cfV\ (cf.~\citealt[Lem.~and Def.~14]{SchroederPattinson10} for the relationship with neighbourhood semantics).  A result of Schr\"oder \citeyearpar[Thm.~31]{Schroeder2008} shows that complete additivity of the associated modal operator is only available for  structures collapsing to a special case of Kripkeanity. This by itself does not lead to ``natural'' examples of coalgebraic logics where mild rules like our \ref{V-mod} are inadmissible; coalgebraic logicians focus almost exclusively on minimal logics,  axiomatizable by formulas of \emph{rank 1} \citep{SchroederPattinson10,Kurz2012} where the issue of incompleteness simply does not arise. However, in extended formalisms like the recently proposed \emph{Coalgebraic Predicate Logic} ($\nml{CPL}$, \citealt{Litak2012,LitakPSS18}), one can write \emph{sentences} capturing complete additivity of a modality and show completeness for theories obtained this way. These sentences essentially involve equality in the same way as capturing \cfA\cfV\ requires nominals or the \emph{difference modality} \citep{Litak2006}. Can we use the syntactic investigations of \S~\ref{sec:newsyntax} to formulate completeness results for Kripkean instances of  $\nml{CPL}$ in an equality-free language?} 

Fifth, we would like to mention two problems brought up by participants of ToLo~V in Tbilisi, where this work was presented in June 2016. Do there exist natural properties of \baos\ yielding notions of completeness properly intermediate between \cfV- and \cfT-completeness? And is there a topological way of deriving the first-order characterization of complete additivity for operators on arbitrary posets in \citealt{Andreka2016}?\footnote{We thank Mamuka Jibladze for the first problem and Sam van Gool for the second. In the latter case, we also appreciate email discussions with Mai Gehrke and Marcel Ern\'e. A detailed discussion would take us too far afield, but at the moment there seems to be no obvious answers to these questions.} 

Finally, our present success should be an encouragement to revisit other open problems regarding sub-Kripkean completeness posed by Litak [\citeyear{Litak2004}; \citeyear[Ch. 9]{Litak2005b}; \citeyear[\S~7]{Litak2008}] and others. For example, given the importance of \textit{transitive} modal logics---normal extensions of $\nml{K4}$---it is natural to ask what incompleteness phenomena arise for these logics. We still have no indication of the existence of any \cfA\cfT-incomplete transitive unimodal logics. It might even be the case that van Benthem's \citeyearpar{Benthem1979} weak second-order consequence is conservative for them.  There is also no indication that the Blok Dichotomy will generalize to degrees of \cfA-incompleteness. In fact, as discussed by \cite{Venema2007} and \cite{Litak2008}, a strong result by Buszkowski \citeyearpar{Buszkowski86,Buszkowski04} implying \cfA-completeness of logics axiomatizable by \emph{modally guarded} axioms suggests that examples of \cfA-incompleteness are few and far between (note that the axiom of $\vblog$ \emph{is} modally guarded). Most importantly, we still seem to have very few (if any) general results regarding sub-Kripkean completeness for non-classical logics with a non-Boolean propositional base. Even Kuznetsov's problem \citep{Kuznetsov1975} regarding topological completeness of superintuitionistic logics remains open more than four decades after its formulation (see \citealt{B&H2018}). We hope to see progress on these problems in the years ahead.\\

\noindent \textbf{Acknowledgements}. For helpful comments, we wish to thank Johan van Benthem, Lloyd Humberstone, James Walsh, and the referees for  \textit{The Review of Symbolic Logic}. We are grateful to the Group in Logic and the Methodology of Science at UC Berkeley for funding a visit by Litak to Berkeley. We would also like to salute Erwin R. Catesbeiana for his insistence and persistence  regarding algebraic incompleteness and inconsistency.

\nocite{Tarski1956}
\bibliographystyle{plainnat}
\bibliography{Vincompleteness}
               
\end{document}